\documentclass[reqno, 12pt]{amsart}

\usepackage{amssymb, amsfonts, graphicx, epsf, here, bm, amsmath, amsthm, natbib}
\usepackage{color}
\usepackage{pstricks, pst-plot}
\usepackage{booktabs}

\graphicspath{{Figures/}}

\makeatletter \@addtoreset{equation}{section} \makeatother

\newtheorem{thm}{Theorem}[section]
\newtheorem{lem}[thm]{Lemma}
\newtheorem{prop}[thm]{Proposition}

\newtheoremstyle{example}{\topsep}{\topsep}%
     {}
     {}
     {\bfseries}
     {:}
     { }
     {\thmname{#1}\thmnumber{ #2}\thmnote{ #3}}

\theoremstyle{example}

\newcommand{\RR}{\mathbb R}

\newcommand{\eqd}{\buildrel d\over=}

\usepackage{fancyhdr}
\begin{document}

\title{Squaring the Circle and Cubing the Sphere:  Circular and Spherical Copulas}
\author{Michael D. Perlman and Jon A. Wellner}

\address
{
\newline \noindent Department of  Statistics, University University of Washing
\newline e-mail:  \rm \texttt{michael@stat.washington.edu}
}
\address
{
\newline \noindent Department of  Statistics, University of Washington
\newline e-mail:  \rm \texttt{jaw@stat.washington.edu}
}

\date{\today}

\begin{abstract}

Do there exist circular and spherical copulas in $\mathbb{R}^d$? 
That is, do there exist circularly symmetric distributions on the unit disk in $\mathbb{R}^2$ 
and spherically symmetric distributions on the unit ball in $\mathbb{R}^d$, $d\ge3$,  
whose one-dimensional marginal distributions are uniform? 
The answer is yes for $d=2$ and 3, where the circular and spherical copulas 
are unique and can be determined explicitly, but no for $d\ge4$. 
A one-parameter family of elliptical bivariate copulas is obtained from 
the unique circular copula in $\mathbb{R}^2$ by oblique coordinate 
transformations. Copulas obtained by a non-linear transformation of a 
uniform distribution on the unit ball in  $\mathbb{R}^d$ are also described, 
and determined explicitly for $d=2$.
\vskip 0.4cm

\end{abstract}

\maketitle

\section{Introduction}
\label{sec:intro}

\noindent  Do there exist spherically symmetric distributions on the closed unit ball 
$B_d$ in $\mathbb{R}^d$ that have uniform one-dimensional marginal distributions on $[-1,1] $? 
A distribution on $B_d$ with this property may be said to ``square the circle" 
when $d=2$ and to ``cube the sphere" when $d\ge3$.

The cumulative distribution function (cdf) of a multivariate distribution 
on the unit cube $[0,1]^d$ whose marginal distributions are uniform[0,1] is
commonly called a {\it copula}; see 
\citet{MR2197664} 
for an accessible introduction to this topic. However, although it is 
customary to confine attention to distributions on the {\it unit} cube, 
our interest is in \emph{spherically symmetric (= orthogonally invariant) 
distributions on $B_d$ with uniform marginal distributions.} 
Therefore we take ``copula" to mean a multivariate cdf on the 
{\it centered} cube $C_d:=[-1,1]^d$ with uniform[-1,1] marginals.  

For $d=2$ (resp., $d\ge3$), such a copula, if it exists, will be called a \emph{circular copula (resp., spherical copula)} if it is the cdf of a circularly symmetric (resp., spherically symmetric) distribution on the unit disk $B_2$ (resp., unit ball $B_d$). 

It will be noted in Sections~\ref{sec:uniqueexist} and ~\ref{sec:circularcopula} 
that circular and spherical copulas are unique if they exist, but exist only for 
dimensions $d=2$ and $d=3$. The proof of non-existence for $d\ge4$ is remarkably simple. 
Explicit expressions for these copulas are given in Sections \ref{sec:circularcopula}   
and \ref{sec:sphericalcopula} respectively. 

In Section \ref{sec:ellipticalcopula}, a new one-parameter family of bivariate copulas 
called {\sl elliptical copulas} is 
obtained from the unique circular copula in $\mathbb{R}^2$ by oblique coordinate transformations. 
Finally, in Section \ref{sec:nonlinearcopula}, copulas obtained by a non-linear 
transformation of a uniform distribution on the unit ball in $\mathbb{R}^d$ are described, 
and determined explicitly for $d=2$.

\section{Uniqueness and Existence of Circular and Spherical Copulas}
\label{sec:uniqueexist}

\begin{prop}\label{uniqueness}
Circular and spherical copulas are unique if they exist.
\footnote{This result is well-known (e.g., 
\citet{MR0270403},  
pp.31-33, 
who uses ``random direction" to indicate the uniform distribution of 
$U\in\partial B_3$), and reappears frequently (e.g. 
\citet{MR1857936},  
Theorem 3.1). 
The essence of the result goes back at least to 
\citet{MR1503439}.}  
\end{prop}
\begin{proof}
If a circular or spherical copula exists on $C_d$, it is the cdf of a random 
vector $Z\equiv(Z_1,\dots,Z_d)$  with a spherically symmetric distribution 
on $B_d$ and with each $Z_i\sim\mathrm{uniform}[-1,1]$. 
The latter implies that $Z$ has no atom at the origin, i.e., $P[Z=0]=0$, 
so we may consider the  ``polar coordinates" representation $Z = R\cdot U$, 
where $R=\|Z\|\le1$ and $U=Z/\|Z\|$. It is well known (e.g., 
\citet{MR629795}, 
Lemmas 1 and 2) that the random unit vector $U\equiv(U_1,\dots,U_d)$ is independent 
of $R$ and is uniformly distributed on the unit sphere $\partial B_d$, 
which implies that each $U_i^2\sim\mathrm{Beta}(1/2,\,(d-1)/2)$. 
Since $Z_i=RU_i$, we have that 
\begin{equation}\label{logZRU}
\log(Z_i^2)=\log(R^2)+\log(U_i^2).
\end{equation}
Because $R$ and $U_i$ are independent, it follows that the characteristic function 
of $\log(R^2)$ is the quotient of the characteristic functions of $\log(Z_i^2)$ and $\log(U_1^2)$. 
Thus the distribution of $\log(R^2)$, and therefore that of $R$,  is uniquely 
determined by the distributions of $Z_i^2$ and $U_i^2$, which are already specified above. 
Thus the the joint distribution of $(R,U)$ is uniquely determined, 
hence so is the distribution of $Z$, hence so its cdf = copula.
\end{proof}

The existence of spherical copulas is easy to determine in three or more dimensions:

\begin{prop}\label{threeormoredim}
Spherical copulas do not exist for $d\ge4$. For $d=3$, the unique spherical copula 
is generated by the uniform distribution on the unit sphere 
$\partial B_3:=\{(x_1,x_2,x_3)\mid x_1^2+x_2^2+x_3^2=1\}$.
\end{prop}
\begin{proof} Let  $Z$ be as in the proof of Proposition \ref{uniqueness}. Then
\begin{equation}\label{EasyMomentInequality}
\frac{1}{3} = E(Z_i^2) = E (R^2 ) E(U_i^2) \le  \frac{1}{d}.
\end{equation}
since $Z_i\sim\mathrm{uniform}[-1,1]$, $0\le R\le1$, 
and $U_i^2\sim\mathrm{Beta}(1/2,\,(d-1)/2)$. Thus $d\le3$, 
so a spherical copula cannot exist when $d\ge4$.

Furthermore, if a spherical copula is to exist for $d=3$, it follows from 
\eqref{EasyMomentInequality} that its generating random vector $Z\in B_3$ 
must satisfy $E(R^2)=1$, hence $R=1$ with probability one. 
This can occur only if $Z$ is uniformly distributed on the unit sphere 
$\partial B_3$. But it is well known\footnote{This follows from the fact that the area of a spherical 
zone is proportion to its altitude -- cf. 
\citet{MR0270403},  
 Proposition (i), p. 30. } 
that this distribution does indeed have uniform marginal distributions 
on $[-1,1]$, hence generates  the unique spherical copula for $d=3$.
\end{proof}

\section{The Bivariate Case: the Unique Circular Copula}
\label{sec:circularcopula}

The following three questions constitute an engaging classroom exercise.
\medskip

\noindent {\bf Question 1.} Let $(X,Y)$ be a random vector uniformly distributed on the unit 
disk (= ball) $B_2$ in $\mathbb{R}^2$. Find the marginal probability distributions of $X$ and $Y$.
\smallskip

\noindent {\bf Answer 1.} One can easily show that $X$ has 
the ``semi-circular" probability density function (pdf) given
by
\begin{equation}
\label{pdf1}
f(x)=\frac{2}{\pi}\sqrt{1-x^2},\quad -1\le x\le1.
\end{equation}
(See Figure~\ref{fig3.1b}.) 
By symmetry, $Y$ has the same pdf as $X$.
\medskip

\noindent  {\bf Question 2.} Let $(X,Y)$ be a random vector uniformly distributed on the unit
circle $\partial B_2$ in $\mathbb{R}^2$. Find the marginal probability distributions of $X$ and $Y$.
\smallskip

\noindent {\bf Answer 2.} We can represent $(X,Y)$ as $(\cos\Theta,\,\sin\Theta)$ where
 $\Theta\sim{\rm uniform}[0,2\pi)$. It follows readily that $X$ has pdf
\begin{equation}\label{pdf2}
f(x)=\frac{1}{\pi\sqrt{1-x^2}},\quad -1<x<1.
\end{equation}
(See Figure~\ref{fig3.1b}.)  By symmetry, $Y$ has the same pdf as $X$. 

\begin{figure}[h]
\centering
\includegraphics[scale=1]{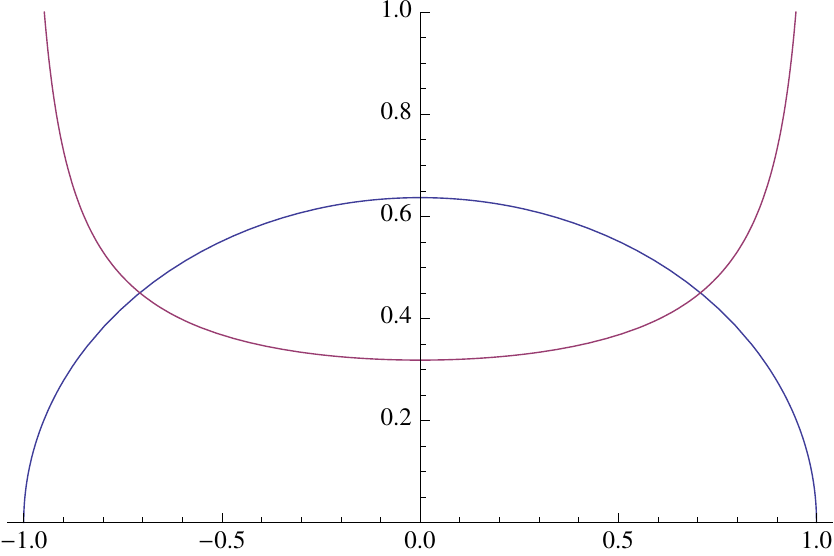}
\caption{The densities (3.1) (lower, blue) and (3.2) (upper, purple).} 
\label{fig3.1b}
\end{figure}

\smallskip

In both cases, the joint distribution of $(X,Y)$ is circularly symmetric, that
is, invariant under all orthogonal transformations of $\mathbb{R}^2$. A comparison of the shapes of
the pdfs in Figure~\ref{fig3.1b} suggest a third question:
\medskip

\noindent  {\bf Question 3.} Does a circularly symmetric bivariate 
distribution with uniform$[-1,1]$  marginals exist on $B_2$? If so, it 
determines a circular copula on $C_2$, which is unique by Proposition \ref{uniqueness}.
This also follows from uniqueness results for the Abel transform; see e.g. \citet{MR924577}.
\medskip
\smallskip

\noindent {\bf Answer 3.} 
Optimistically, let's seek an absolutely continuous solution. 
That is, we seek a bivariate pdf on $B_2$ of the form
\begin{equation}
f(x,y)=g(x^2+y^2)
\nonumber
\end{equation}
such that the marginal pdf
\begin{equation}
f(x)\equiv\int_{-\sqrt{1-x^2}}^{\sqrt{1-x^2}}f(x,y)dy,\qquad -1<x<1,
\nonumber 
\end{equation}
is constant in $x$. Here $g$ is a nonnegative function on $(0,1)$ that must satisfy 
\begin{equation}\label{gcondition}
 2\pi\int_0^1 rg(r^2)dr=1
\end{equation}
in order that $\int\!\!\int_{B_2} f(x,y)dxdy=1$ (transform to polar coordinates:
$(x,y)\to(r,\theta)$). 

To determine a suitable $g$, first set $h(t)=g(1-t)$, then let $u=\frac{y}{\sqrt{1-x^2}}$ to
obtain
\begin{eqnarray*}
f(x)&=&\int_{-\sqrt{1-x^2}}^{\sqrt{1-x^2}}h(1-x^2-y^2)dy\label{marginalx2}\\
 &=&\sqrt{1-x^2}\int_{-1}^{1}h((1-u^2)(1-x^2))du\label{marginalx3}\\
  &=&2\sqrt{1-x^2}\int_{0}^{1}h((1-u^2)(1-x^2))du\label{marginalx4}. \label{hRepresentOfMarginalDensity}
\end{eqnarray*}
If we take $h(t)=c\, t^{-1/2}$ then clearly $f(x)$ does not depend on $x$,
and choosing $c=1/2\pi$ satisfies (\ref{gcondition}). Thus the
bivariate pdf
\begin{equation}\label{fcirc}
f(x,y)=\frac{1}{2\pi\sqrt{1-x^2-y^2}},\qquad x^2+y^2<1,
\end{equation}
determines a circularly symmetric bivariate distribution on $B_2$ and yields the desired circular copula.
\medskip

\begin{figure}[h]
\centering
\includegraphics[scale=1]{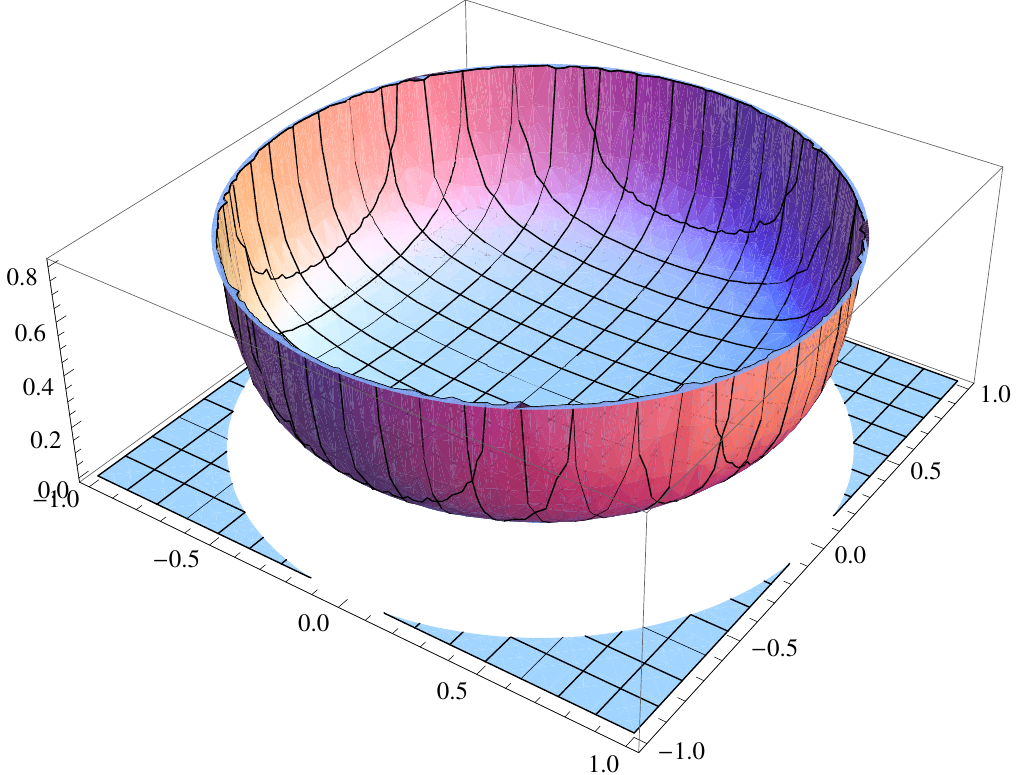}
\caption{Circularly symmetric bivariate density on $B_2$} 
\label{fig3.1}
\end{figure}
    
\noindent{\bf Question 4.} Having determined the unique circularly 
symmetric distribution \eqref{fcirc} on $B_2$ with uniform marginals, 
what is the corresponding cdf $F(x,y)$, that is, what is the corresponding circular copula?
\smallskip

\noindent{\bf Answer 4.} 
The circular symmetry of $(X,Y)$ implies that its distribution 
is invariant under sign changes, i.e., $(X,Y)\eqd(\pm X,\,\pm Y)$. 
By  the following lemma, the cdf $F(x,y)\equiv P[X\le x,\,Y\le y]$ on $C_2\equiv[-1,1]^2$ 
can be expressed in terms of  $F_0(x,y)$, its truncation to the first quadrant:
\begin{equation}\label{F0}
F_0(x,y)\equiv P[0\le X\le x,\,0\le Y\le y]
\end{equation}
for $0\le x,y\le1$, and also in terms of the complementary cdf 
$\bar F(x,y)\equiv P[X>x,\,Y>y]$ for $0\le x,y\le1$. 
Because $(X,Y)\eqd(\pm X,\,\pm Y)$ and  has uniform$[-1,1]$ marginals, 
\begin{eqnarray}
F_0(x,y)&=&P[0\le X\le1,\,0\le Y\le1]-P[X>x,\,0\le Y\le 1]\nonumber\\
 &&-P[0\le X\le1,\,Y>y]+P[X>x,\,Y>y]\nonumber\\
 &=&\frac{1}{4}-\left(\frac{1-x}{4}\right)-\left(\frac{1-y}{4}\right)+\bar F(x,y)\nonumber\\
 &=&\frac{x+y-1}{4}+\bar F(x,y),\quad 0\le x,y\le 1.\label{F0Fbar}
\end{eqnarray}

\begin{lem}\label{absvalues}
Let $(X,Y)$ be a bivariate random vector on $C_2$ with uniform$[-1,1]$ 
marginal distributions and sign-change invariance, i.e., $(X,Y)\eqd(\pm X,\,\pm Y)$. Then for $(x,y)\in C_2$,
\begin{eqnarray}
F(x,y)&=&\frac{x+y+1}{4}+\sigma(xy)\,F_0(|x|,|y|)\label{XYabsXY}\\
&=&\frac{x+y+1}{4}+\sigma(xy)\left[\frac{|x|+|y|-1}{4}+\bar F(|x|,|y|)\right],\label{XYabsXYcomp}
\end{eqnarray}
where $\sigma(w)=\mathrm{sign}(w)$ if $w\ne0$ and $\sigma(0)=0$.
\end{lem}
\begin{proof} To obtain \eqref{XYabsXY}, consider four cases:
\medskip

\noindent{\it Case 1:} $0\le x,y\le1$.  
Because $(X,Y)$ is sign-change invariant and has uniform$[-1,1]$  marginals,
\begin{eqnarray*}
F(x,y)&=&P[0< X\le x,\,0< Y\le y]+P[0< X\le x,\,Y\le0]\\
 &&+P[X\le0,\,0< Y\le y]+P[X\le0,\,Y\le0]\\
&=& F_0(x,y)+\frac{x}{4}+\frac{y}{4}+\frac{1}{4}\\
 &=&\frac{x+y+1}{4}+\sigma(xy)\,F_0(|x|,|y|).
\end{eqnarray*}
\smallskip

\noindent{\it Case 2:} $-1\le x\le0\le y\le1$. Similarly,
\begin{eqnarray*}
F(x,y)&=&\ \ P[X\le0,\,0< Y\le y]-P[x< X\le 0,\,0\le Y\le y]\\
 &&+P[X\le0,\,Y\le0]-P[x< X\le 0,\,Y\le 0]\\
&=&\frac{y}{4}- F_0(-x,y)+\frac{1}{4}-\frac{(-x)}{4}\\
 &=&\frac{x+y+1}{4}+\sigma(xy)\,F_0(|x|,|y|).
\end{eqnarray*}
\smallskip

\noindent{\it Case 3:} $-1\le y\le0\le x\le1$. Similarly,
\begin{eqnarray*}
F(x,y)&=&P[0< X\le x,\,Y\le0]-P[0\le X\le x,\,y< Y\le 0]\\
 &&+P[X\le0,\,Y\le0]-P[X\le 0,\,y< Y\le 0]\\
&=&\frac{x}{4}-F_0(x,-y)+\frac{1}{4}-\frac{(-y)}{4}\\
&=&\frac{x+y+1}{4}+\sigma(xy)\,F_0(|x|,|y|).
\end{eqnarray*}
\smallskip

\noindent{\it Case 4:} $-1\le x,y\le0$. Similarly,
\begin{eqnarray*}
F(x,y&)=&P[X\le 0,\,Y\le 0]-P[x<X\le0,\,Y\le 0]\\
 &&-P[X\le 0,\,y<Y\le0]+P[x<X\le0,\,y<Y\le0]\\
&=&\frac{1}{4}-\frac{(-x)}{4}-\frac{(-y)}{4}+ F_0(-x,-y)\\
&=&\frac{x+y+1}{4}+\sigma(xy)\,F_0(|x|,|y|).
\end{eqnarray*}
\medskip

Finally, \eqref{XYabsXYcomp} follows from \eqref{XYabsXY} by \eqref{F0Fbar}.
\end{proof}
\medskip

Thus, to determine the circular copula $F(x,y)$ for the pdf \eqref{fcirc}, 
it suffices to determine the complementary cdf $\bar F(x,y)$ for $0\le x,y\le1$ 
and apply \eqref{XYabsXYcomp}. Because $\bar F(x,y)=0$ when  $x^2+y^2\ge1$, 
we need only consider the case where $x^2+y^2<1$.
\medskip

\noindent{\it First approach:} 
When $0\le x,y\le 1$ and $x^2+y^2<1$, $\bar F(x,y)$ can be expressed as follows. 
By using Figure~\ref{fig3.2} we find that

\begin{figure}[h]
\centering
\includegraphics[scale=.9]{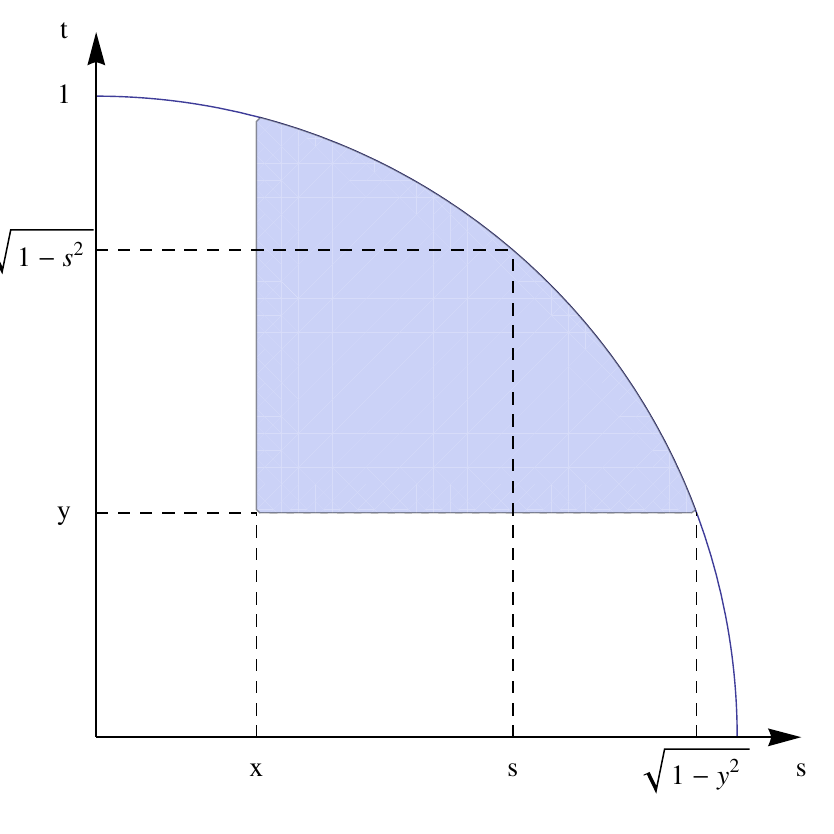}
\caption{Region of integration, 2-dimensional case} 
\label{fig3.2}
\end{figure}

\begin{eqnarray}
\bar F(x,y)
&=&\frac{1}{2\pi}\int_{x}^{\sqrt{1-y^2}} \left \{ \int_{y}^{\sqrt{1-s^2}}\frac{1}{\sqrt{1-s^2-t^2}}dt\right \}ds \nonumber\\
  &=&\frac{1}{2\pi}\int_{x}^{\sqrt{1-y^2}} \left \{ \int_{y}^{\sqrt{1-s^2}}\frac{1}{\sqrt{(1-\frac{t^2}{1-s^2})}}\frac{dt}{\sqrt{1-s^2}}
                         \right \} ds \nonumber\\
  &=&\frac{1}{2\pi}\int_x^{\sqrt{1-y^2}} \left \{ \int_{\frac{y}{\sqrt{1-s^2}}}^1 \frac{dv}{\sqrt{1-v^2}}\right \} ds \nonumber\\
   &=&\frac{1}{2\pi}\int_x^{\sqrt{1-y^2}}\left[\frac{\pi}{2}-\arcsin\left(\frac{y}{\sqrt{1-s^2}}\right)\right]ds.\label{Fbarxysin}
\end{eqnarray}
However, we were unable to evaluate this integral directly. 
\medskip

\noindent{\it Second approach:} 
Fortunately, we have found a solution in the molecular biology and optics literatures, 
where the problem of finding the area of the intersection of two spherical caps on the 
unit sphere $\partial B_3$ has been addressed. 
The following general result is due to \citet{TovchiVakser:01} and also appears in 
\citet{DBLP:conf/si3d/2007}.

\begin{lem}\label{surfacearea}
Let $S_1$ and $S_2$ be spherical caps on $\partial B_3$. Let $r_1$ and $r_2$ 
denote their angular radii and let $d$ denote the angular distance between their 
centers ($0<d\le\pi$). Assume that $0< r_1, r_2\le \pi/2$ and $d\le  r_1+ r_2$, 
so that the intersection $S_1\cap S_2\ne\emptyset$ and consists of a single ``diangle"; 
(see Figures~\ref{fig3.4} and ~\ref{fig3.5}.)
Then Area$(S_1\cap S_2)$ is given by
\begin{eqnarray}
A(r_1,r_2;d)
& = & 2 \pi - 2\pi \cos(r_1) - 2\pi \cos (r_2) 
           - 2 \arccos \left (\frac{\cos(d)-\cos(r_1)\cos(r_2)}{\sin(r_1) \sin (r_2)} \right )\nonumber \\
&& + \ 2 \cos(r_1) \arccos \left ( \frac{\cos(d)\cos(r_1)-\cos(r_2)}{\sin(d)\sin(r_1)} \right )\nonumber\\
&&  + \ 2 \cos(r_2) \arccos \left ( \frac{\cos(d)\cos(r_2)-\cos(r_1)}{\sin(d)\sin(r_2)} \right ).\label{areacapintersect}
\end{eqnarray}
\end{lem}

This result can be applied to obtain our desired circular copula as follows.
\medskip

If $(X,Y,Z)$ is uniformly distributed on $\partial B_3$, then the event $\{X>x,\,Y> y\}$ 
corresponds to the intersection of the two spherical caps 
$\{X>x\}$ and $\{Y>y\}$, so $P[X>x,\,Y> y]$ is given by the area $A(x,y)$ of this 
intersection divided by the total area of $\partial B_3$, i.e., by $4\pi$. 
(See Figures~\ref{fig3.4} and ~\ref{fig3.5}.)  
Also, the joint distribution of $(X,Y)$ is circularly symmetric on the unit disk $B_2$ 
and has uniform marginals, so must be the unique such bivariate distribution, 
namely the distribution with pdf \eqref{fcirc}. 

\begin{figure}
\centering
\includegraphics[scale=1]{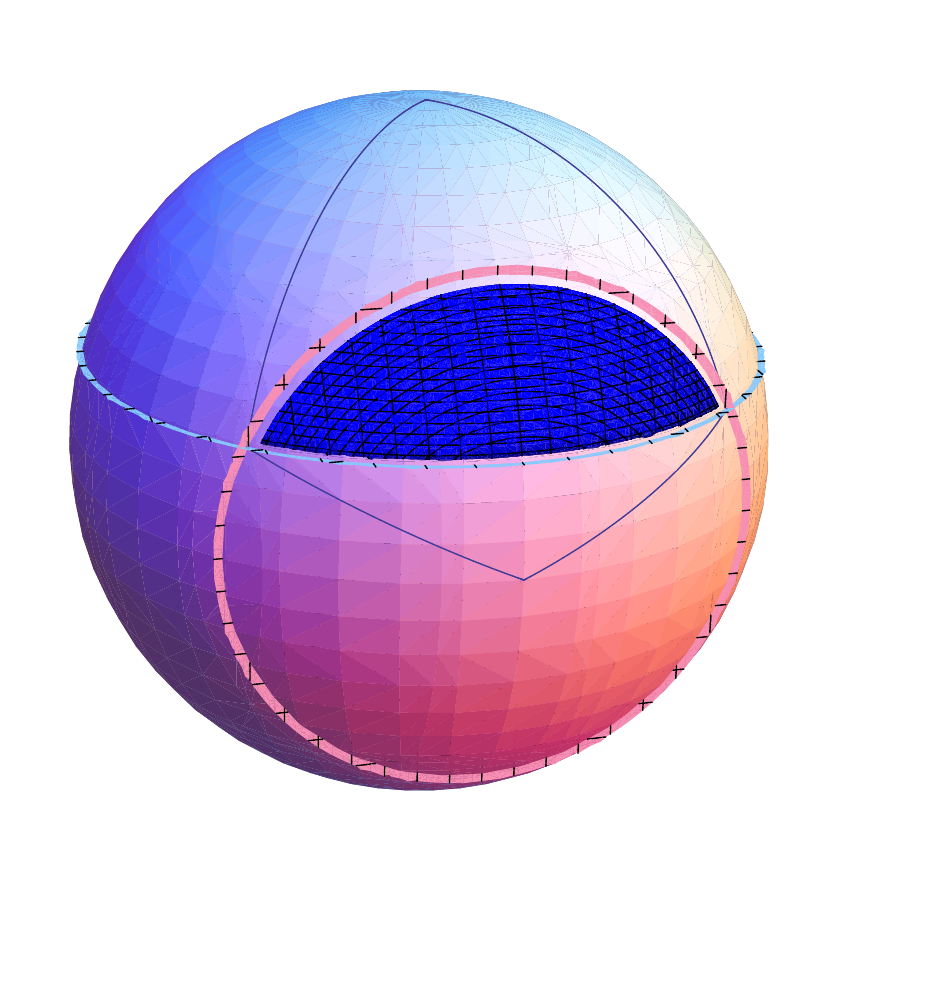}
\caption{Intersection of two spherical caps} 
\label{fig3.4}
\end{figure}

\begin{figure}
\centering
\includegraphics[scale=1]{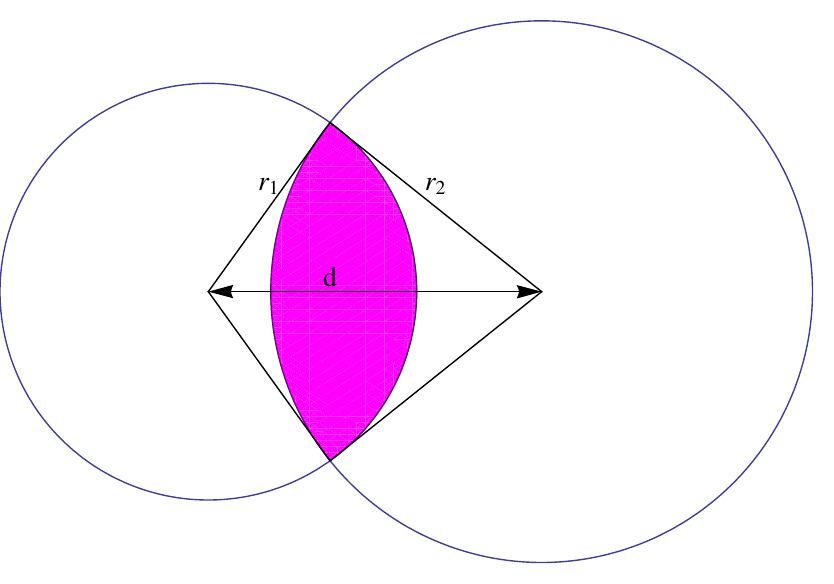}
\caption{Intersection of two spherical caps, 
circular representation (modified from \citet{TovchiVakser:01})} 
\label{fig3.5}
\end{figure}

Thus, for $0\le x,y\le 1$ and $x^2+y^2<1$, our desired complementary cdf is given by 
\begin{eqnarray}
\bar F(x,y)&=&\frac{1}{4\pi}A(x,y)\label{FbarAxy}\\
&=& \frac{1}{4\pi}A(\arccos(x),\,\arccos(y);\,\pi/2)\label{AA}\\
& = & \frac{1}{2} - \frac{x}{2} - \frac{y}{2} 
          - \frac{1}{2\pi} \arccos \left ( - \frac{xy}{ \sqrt{(1-x^2)(1-y^2)}} \right )\nonumber\\
&&  + \ \frac{x}{2\pi} \arccos \left ( \frac{- y}{\sqrt{1-x^2}} \right ) 
       + \frac{y}{2\pi} \arccos \left ( \frac{-x}{\sqrt{1-y^2}} \right )\nonumber\\
 &\equiv& \frac{1-x-y}{4}+\alpha(x,y),\label{psidef}
\end{eqnarray}
where for $0\le x,y\le1$  and $x^2+y^2<1$,
\begin{eqnarray}
\alpha(x,y)
&=& \frac{1}{2\pi}\left[x\arcsin \left ( \frac{y}{\sqrt{1-x^2}} \right ) 
           +y\arcsin\left ( \frac{x}{\sqrt{1-y^2}} \right )\right.\nonumber\\
&& \left.\qquad-\arcsin\left ( \frac{x y}{ \sqrt{(1-x^2)(1-y^2)}} \right )\right].\label{psi}
\end{eqnarray}
\begin{thm}\label{circularcopula}
The unique circular copula on $C_2$ is given by
\begin{equation}\label{circcopformula}
F(x,y)=\frac{x+y+1}{4}+\alpha(x,y),
\end{equation}
where $\alpha(x,y)$ is defined by \eqref{psi} for $x^2+y^2<1$ and by
\begin{equation}\label{psiextension}
\alpha(x,y)=\sigma(xy)\cdot\left(\frac{|x|+|y|-1}{4}\right)
\end{equation}
 for $x^2+y^2\ge1$. Note that \eqref{psi} and \eqref{psiextension} 
 agree when $x^2+y^2=1$ and both are sign-change equivariant on 
 $C_2$: for all $(x,y)\in C_2$ and all $\epsilon,\delta=\pm1$,
\begin{equation}\label{psiantisymmetry} 
\alpha(\epsilon x,\delta y)=\epsilon \delta\cdot\alpha(x,y).
\end{equation}
\end{thm}
\begin{proof}
From \eqref{XYabsXYcomp} and \eqref{psidef}, when $x^2+y^2<1$ we have
\begin{eqnarray}
F(x,y)&=&\frac{x+y+1}{4}+\sigma(xy)\,\alpha(|x|,|y|)\\
 &=&\frac{x+y+1}{4}+\alpha(x,y)
\end{eqnarray}
by \eqref{psiantisymmetry}. When $x^2+y^2\ge1$, $\bar F(|x|,|y|)=0$ so 
\eqref{circcopformula} again holds by \eqref{XYabsXYcomp} and \eqref{psiextension}.
\end{proof}

See Figure~\ref{fig3.3} for
a plot of the resulting copula (on $[-1,1]^2$).

\begin{figure}[h]
\centering
\includegraphics[scale=1]{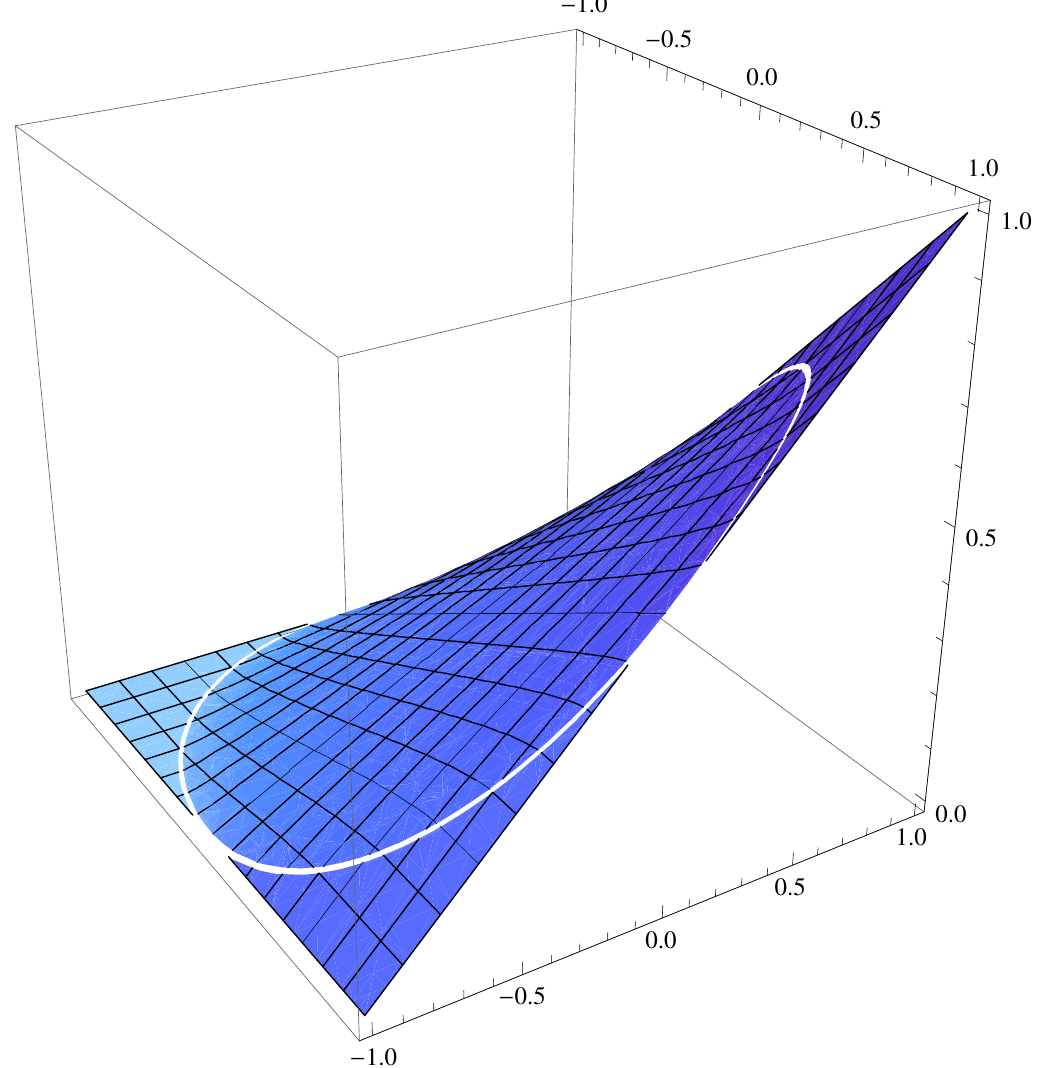}
\caption{Copula (joint distribution function), Theorem 3.3} 
\label{fig3.3}
\end{figure}

\medskip

\section{The Trivariate Case: the Unique Spherical Copula}
\label{sec:sphericalcopula}

\noindent{\bf Question 5.} 
Having determined the unique spherically symmetric distribution on $B_3$ 
with uniform marginals, namely, the uniform distribution on the unit sphere $\partial B_3$, 
what is the corresponding cdf $F(x,y,z)$ on $C_3$, i.e,  the unique spherical copula?
\medskip

\noindent{\bf Answer 5.} 
As in Section \ref{sec:circularcopula},
let $(X,Y,Z)$ be uniformly distributed on $\partial B_3$, so that $F(x,y,z)= P[X\le x,\,Y\le y,\,Z\le z]$. 
Again we first determine the complementary cdf 
$\bar F(x,y,z)\equiv P[X>x,\,Y>y,\,Z>z]$ for $0\le x,y,z\le1$ and $x^2+y^2+z^2<1$, 
the intersection of the first octant of $C_3$ with the interior of $B_3$.  
Here the event $\{X>x,\,Y>y,\,Z>z\}$ corresponds to the intersection of the three 
spherical caps $\{X>x\}$, $\{Y>y\}$, and $\{Z>z\}$ on $\partial B_3$, 
so $\bar F(x,y,z)$ is the area $A(x,y,z)$ of this intersection divided by the total area $4\pi$ of $\partial B_3$.
\smallskip

Recall that two approaches were proposed in Section \ref{sec:circularcopula}
to obtain the area $A(x,y)$ of the intersection of {\it two} circular caps $\{X>x\}$ and $\{Y>y\}$. 
The first approach led to the integral \eqref{Fbarxysin} that we were unable to evaluate explicitly, 
so we adopted a second approach based on the geometric Lemma \ref{surfacearea} of \citet{TovchiVakser:01}. 
Andrey Tovchigrechko has kindly suggested a method for extending Lemma \ref{surfacearea} 
to the case of three spherical caps in general position, which if carried out would yield an explicit expression for $A(x,y,z)$. 
However, we have found that because the axes of our three caps are mutually orthogonal, 
the two approaches just mentioned for the bivariate case can be combined  
to obtain $\bar F(x,y,z)\equiv\frac{1}{4\pi}A(x,y,z)$ directly for the trivariate case, as now described. 
\smallskip

We begin by extending \eqref{Fbarxysin} to obtain an integral expression for 
$\bar F(x,y,z)$ when $0\le x,y,z\le1$ and $x^2+y^2+z^2<1$. We require the fact that
\begin{equation}\label{arcsinab}
0 \le a,b \le 1\ \  \mathrm{and}\ \ a^2+b^2=1  \ \ \mbox{implies} \ \  \arcsin(a) + \arcsin(b) = \pi/2.
\end{equation}

\begin{lem} If $0\le x,y,z\le1$ and $x^2+y^2+z^2<1$, then
\begin{equation}\label{integralFbarxyz}
\bar F(x,y,z)
   =\frac{1}{4\pi}\int_x^{\sqrt{1-y^2-z^2}}\left[\frac{\pi}{2}
     -\arcsin\left(\frac{y}{\sqrt{1-s^2}}\right)-\arcsin\left(\frac{z}{\sqrt{1-s^2}}\right)\right]ds.
\end{equation}
Because $(X,Y,Z)$ is exchangeable, \eqref{integralFbarxyz} remains valid 
under any permutation of $x,y,z$ on the right-hand side.
\end{lem}
\begin{proof}
Since $X^2+Y^2+Z^2=1$ and $(X,Y,Z)\eqd(X,Y,-Z)$, 
it follows from \eqref{fcirc} by using Figure~\ref{fig4.1} that 

\begin{figure}[h]
\centering
\includegraphics[scale=1.2]{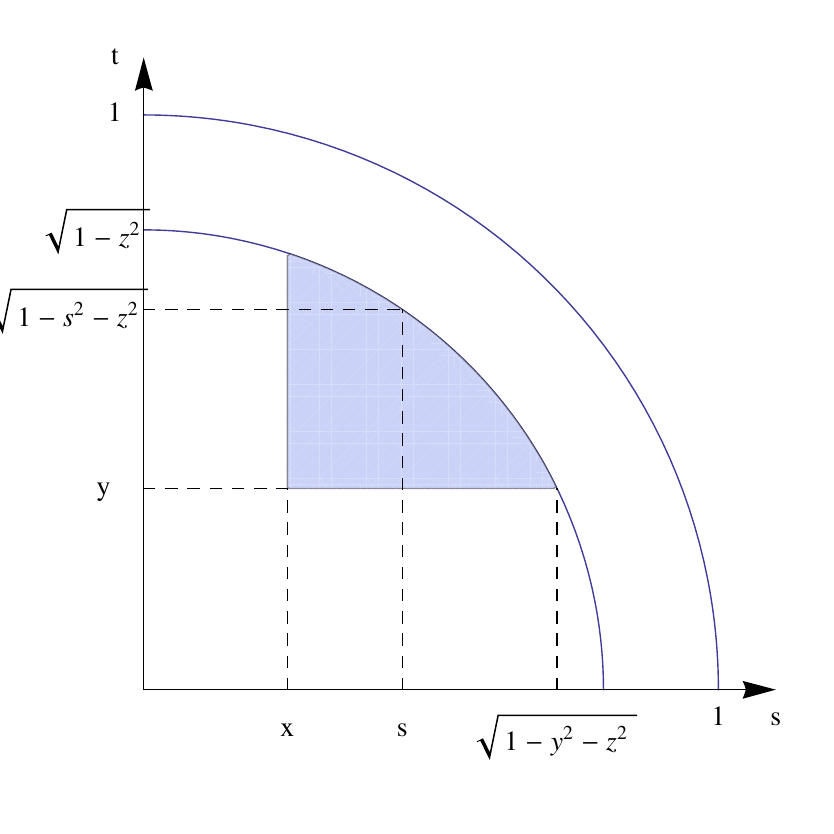}
\caption{Region of integration, 3-dimensional case, Lemma 4.1} 
\label{fig4.1}
\end{figure}

\begin{eqnarray*}
\lefteqn{P[X>x,\,Y>y,\,Z>z]}\\
&=&\frac{1}{2}P[X>x,\,Y>y,\,X^2+Y^2<1-z^2]\\
&=&\frac{1}{4\pi}\int_{x}^{\sqrt{1-y^2-z^2}} \left \{ \int_{y}^{\sqrt{1-s^2-z^2}}\frac{1}{\sqrt{1-s^2-t^2}}dt\right \} ds \nonumber\\
  &=&\frac{1}{4\pi}\int_{x}^{\sqrt{1-y^2-z^2}} 
            \left \{ \int_{y}^{\sqrt{1-s^2-z^2}}\frac{1}{\sqrt{(1-\frac{t^2}{1-s^2})}}\frac{dt}{\sqrt{1-s^2}}\right \} ds \nonumber\\
  &=&\frac{1}{4\pi}\int_x^{\sqrt{1-y^2-z^2}} 
           \left \{ \int_{\frac{y}{\sqrt{1-s^2}}}^{\sqrt{\frac{1-s^2-z^2}{1-s^2}}} \frac{dv}{\sqrt{1-v^2}}\right \} ds \nonumber\\
   &=&\frac{1}{4\pi}\int_x^{\sqrt{1-y^2-z^2}}
           \left[\arcsin\left(\sqrt{\frac{1-s^2-z^2}{1-s^2}}\right)-\arcsin\left(\frac{y}{\sqrt{1-s^2}}\right)\right]ds.
   \end{eqnarray*}
Now apply \eqref{arcsinab} to obtain \eqref{integralFbarxyz}.
\end{proof}

As noted above, the integral in \eqref{integralFbarxyz} appears difficult to evaluate explicitly, 
but the following indirect argument succeeds.
Recall from \eqref{Fbarxysin} and \eqref{psidef} that when $0\le x,y\le 1$ and $x^2+y^2<1$,
\begin{eqnarray*}
\bar{F} (x,y) & = & \frac{1}{2\pi} \int_x^{\sqrt{1-y^2}} \left[ \frac{\pi}{2} - \arcsin\left ( \frac{y}{\sqrt{1-s^2}} \right ) \right] ds \\
& = & \frac{1-x-y}{4}  + \alpha (x,y),
\end{eqnarray*}
where $\alpha(x,y)$ is given by \eqref{psi}. Because $z\le\sqrt{1-x^2-y^2}\le\sqrt{1-y^2}$ when $0\le x,y,z\le1$ and $x^2+y^2+z^2<1$, it follows that 
\begin{eqnarray}
\lefteqn{ \frac{1}{2\pi} \int_z^{\sqrt{1-x^2 - y^2}}  \left[ \frac{\pi}{2} - \arcsin\left ( \frac{y}{\sqrt{1-s^2}} \right ) \right]ds}\nonumber \\
& = & \frac{\sqrt{1-x^2 - y^2}-z}{4} + \alpha (z,y)
  - \alpha (\sqrt{1-x^2 - y^2},\,y).\label{badifference}
\end{eqnarray}
Therefore from \eqref{integralFbarxyz} and \eqref{badifference}, if $0\le x,y,z\le1$ and $x^2+y^2+z^2<1$ then
\begin{eqnarray*}
\lefteqn{4\pi \bar F(x,y,z)} \\
& = &  \int_z^{\sqrt{1-x^2 - y^2}} \left[ \frac{\pi}{2} - \arcsin \left ( \frac{x}{\sqrt{1-s^2}} \right ) \right] ds \\
&&+  \int_z^{\sqrt{1-x^2 - y^2}} \left[ \frac{\pi}{2} - \arcsin \left ( \frac{y}{\sqrt{1-s^2}} \right )  \right] ds
 - \frac{\pi}{2} \left[ \sqrt{1-x^2-y^2} -z \right] \\
& = & \frac{\pi}{2} ( \sqrt{1-x^2-y^2} -z) +2\pi \left[  \alpha (z,x)-\alpha( \sqrt{1-x^2-y^2},\,x)  \right] \\
&&+ \frac{\pi}{2} ( \sqrt{1-x^2-y^2} -z) +2\pi \left[  \alpha (z,y)-\alpha( \sqrt{1-x^2-y^2},\,y)  \right]  \\
&& - \frac{\pi}{2} ( \sqrt{1-x^2-y^2} -z)   \\
& = & \frac{\pi}{2} (\sqrt{1-x^2-y^2} -z) \\
 &&+ 2\pi \left[ \alpha (z,x) + \alpha (z,y) - \alpha (\sqrt{1-x^2-y^2},\, x)   - \alpha (\sqrt{1-x^2-y^2},\, y) \right]\\
& = & \frac{\pi}{2} ( \sqrt{1-x^2-y^2} - z)  \\
&&+\,x \arcsin \left ( \frac{z}{\sqrt{1-x^2}} \right ) 
           + z \arcsin \left ( \frac{x}{\sqrt{1-z^2}} \right ) 
           - \arcsin \left ( \frac{xz}{ \sqrt{(1-x^2)(1-z^2)}} \right ) \\
&& +\,y \arcsin \left ( \frac{z}{\sqrt{1-y^2}} \right ) 
           + z \arcsin \left ( \frac{y}{\sqrt{1-z^2}} \right ) 
           - \arcsin \left ( \frac{yz}{ \sqrt{(1-y^2)(1-z^2)}} \right ) \\
&& -\,x \arcsin \left ( \frac{\sqrt{1-x^2-y^2}}{\sqrt{1-x^2}} \right ) 
          - \sqrt{1-x^2-y^2} \arcsin \left ( \frac{x}{\sqrt{x^2+y^2}} \right ) \\
 && + \arcsin \left ( \frac{x\sqrt{1-x^2-y^2}}{ \sqrt{(1-x^2)(x^2+y^2)}} \right )
  -y \arcsin \left ( \frac{\sqrt{1-x^2-y^2}}{\sqrt{1-y^2}} \right ) \\
  &&  - \sqrt{1-x^2-y^2} \arcsin \left ( \frac{y}{\sqrt{x^2+y^2}} \right ) 
    +\arcsin \left ( \frac{y\sqrt{1-x^2-y^2}}{ \sqrt{(1-y^2)(x^2+y^2)}} \right ) .
\end{eqnarray*}
By \eqref{arcsinab}, however,
\begin{eqnarray*}
\sqrt{1-x^2-y^2} \arcsin \left ( \frac{x}{\sqrt{x^2+y^2}} \right ) &+& \sqrt{1-x^2-y^2} \arcsin \left ( \frac{y}{\sqrt{x^2+y^2}} \right ) \\
&=&\sqrt{1-x^2-y^2} \left(\frac{\pi}{2}\right),
\end{eqnarray*}
so if we define $h(x,y)$ by 
\begin{eqnarray}
h(x,y)
&=&\arcsin\left ( \frac{x y}{ \sqrt{(1-x^2)(1-y^2)}} \right )
          + \arcsin \left ( \frac{x\sqrt{1-x^2-y^2}}{\sqrt{1-x^2}\sqrt{x^2+y^2}} \right )
           \nonumber \\
 && + \arcsin \left ( \frac{y\sqrt{1-x^2-y^2}}{\sqrt{1-y^2}\sqrt{x^2+y^2}} \right )  \label{hxydef} 
\end{eqnarray}
for $0\le x,y\le1$ and $x^2+y^2<1$, then
\begin{eqnarray*}
\lefteqn{4 \pi \bar F(x,y,z)}\\
& = & -\frac{\pi}{2}  z + h(x,y)  - \arcsin\left ( \frac{xy}{\sqrt{(1-x^2)(1-y^2)}} \right )  \\
&& +\ x \arcsin \left ( \frac{z}{\sqrt{1-x^2}} \right ) 
           + z \arcsin \left ( \frac{x}{\sqrt{1-z^2}} \right ) 
           - \arcsin \left ( \frac{xz}{ \sqrt{(1-x^2)(1-z^2)}} \right ) \\
&&+\ y \arcsin \left ( \frac{z}{\sqrt{1-y^2}} \right ) 
           + z \arcsin \left ( \frac{y}{\sqrt{1-z^2}} \right ) 
           - \arcsin \left ( \frac{yz}{ \sqrt{(1-y^2)(1-z^2)}} \right ) \\
&&-\ x \arcsin \left ( \frac{\sqrt{1-x^2-y^2}}{\sqrt{1-x^2}} \right ) 
      -y \arcsin \left ( \frac{\sqrt{1-x^2-y^2}}{\sqrt{1-y^2}} \right ) \\
& = &  - \frac{\pi}{2}  z + h(x,y) + \alpha (x,y) \\
&&  - \ x \arcsin \left ( \frac{y}{\sqrt{1-x^2}} \right ) - y \arcsin \left ( \frac{x}{\sqrt{1-y^2}} \right )\\
&& + \ x \arcsin \left ( \frac{z}{\sqrt{1-x^2}} \right ) 
           + z \arcsin \left ( \frac{x}{\sqrt{1-z^2}} \right ) 
           - \arcsin \left ( \frac{xz}{ \sqrt{(1-x^2)(1-z^2)}} \right ) \\
&& + \ y \arcsin \left ( \frac{z}{\sqrt{1-y^2}} \right ) 
           + z \arcsin \left ( \frac{y}{\sqrt{1-z^2}} \right ) 
           - \arcsin \left ( \frac{yz}{ \sqrt{(1-y^2)(1-z^2)}} \right ) \\
&& - \ x \arcsin \left ( \frac{\sqrt{1-x^2-y^2}}{\sqrt{1-x^2}} \right )
       - y \arcsin \left ( \frac{\sqrt{1-x^2-y^2}}{\sqrt{1-y^2}} \right),
\end{eqnarray*}
where $\alpha(x,y)$ is given by \eqref{psi}. Now \eqref{arcsinab} gives
\begin{eqnarray*}
&&x \arcsin \left ( \frac{y}{\sqrt{1-x^2}} \right ) + x \arcsin \left ( \frac{\sqrt{1-x^2-y^2}}{\sqrt{1-x^2}} \right ) = x \left(\frac{\pi}{2}\right) \\
&& y \arcsin \left ( \frac{x}{\sqrt{1-y^2}} \right ) + y \arcsin \left ( \frac{\sqrt{1-x^2-y^2}}{\sqrt{1-y^2}} \right )  =  y \left(\frac{\pi}{2}\right),
\end{eqnarray*}
so the above simplifies to 
\begin{equation}\label{FBarFirstFormula}
4 \pi \bar F (x,y,z)= - \frac{\pi}{2}  (x+y+z) + h(x,y) + 2\pi\Delta(x,y,z),
\end{equation}
where
\begin{equation}\label{Delta}
\Delta(x,y,z)=\alpha (x,y) + \alpha (x,z) + \alpha (y,z),
\end{equation}
a symmetric function of $(x,y,z)$. By \eqref{arcsinab}, however,
\begin{eqnarray*}
h(x,y)
&=&\frac{\pi}{2}-\arcsin\left ( \frac{\sqrt{1-x^2-y^2}}{ \sqrt{(1-x^2)(1-y^2)}} \right )
          +\arcsin \left ( \frac{x\sqrt{1-x^2-y^2}}{\sqrt{1-x^2}\sqrt{x^2+y^2}} \right )\\         
 &&  +\arcsin \left ( \frac{y\sqrt{1-x^2-y^2}}{\sqrt{1-y^2}\sqrt{x^2+y^2}} \right ) \\
 &\equiv&\frac{\pi}{2}-\gamma+\alpha+\beta,
\end{eqnarray*}
and 
\begin{eqnarray*}
\lefteqn{\sin(\alpha+\beta)}\\
 &=&\sin\alpha\cos\beta+\cos\alpha\sin\beta\\
 &=& \frac{x\sqrt{1-x^2-y^2}}{\sqrt{1-x^2}\sqrt{x^2+y^2}}
         \frac{x}{\sqrt{1-y^2}\sqrt{x^2+y^2}}+\frac{y}{\sqrt{1-x^2}\sqrt{x^2+y^2}}
         \frac{y\sqrt{1-x^2-y^2}}{\sqrt{1-y^2}\sqrt{x^2+y^2}}\\
  &=&\frac{\sqrt{1-x^2-y^2}}{ \sqrt{(1-x^2)(1-y^2)}} \\
  &=&\sin\gamma,
\end{eqnarray*}
so $\alpha+\beta=\gamma$, hence $h(x,y)\equiv\frac{\pi}{2}$ identically in $(x,y)$. Therefore we conclude that
\begin{equation}\label{Fbarxyznice}
\bar F(x,y,z)=\frac{1- x-y-z}{8}  + \frac{\Delta(x,y,z)}{2}
\end{equation}
for $0\le x,y,z\le1$ and $x^2+y^2+z^2<1$.
\smallskip

We now apply \eqref{Fbarxyznice} to obtain the cdf $F(x,y,z)$ for all $(x,y,z)\in C_3$. 
For this, extend the definition of $\Delta$ in \eqref{Delta} to all $(x,y,z)\in C_3$ 
by means of \eqref{psi} and \eqref{psiextension}.
\begin{thm}\label{uniquesphericalcopula}
The unique spherical copula $F(x,y,z)$ on $C_3$ is given as follows:
\smallskip

\noindent for $x^2+y^2+z^2<1$,
\begin{equation*}
F(x,y,z)=
\begin{cases}
\frac{1+x+y+z}{8}+\frac{\Delta(x,y,z)}{2},&\mathrm{if}\ x^2+y^2+z^2<1;\\
\frac{1+x+y+z}{8}+\frac{\Delta(x,y,z)}{2}&\\
+\ \sigma(xyz)\left[\frac{1-|x|-|y|-|z|}{8}+\frac{\Delta(|x|,|y|,|z|)}{2}\right],&\mathrm{if}\ x^2+y^2+z^2\ge1.
\end{cases}
\end{equation*}
\end{thm}

\begin{proof} 
We use repeatedly the facts that 
$(X,Y,Z)\eqd(\pm X,\pm Y, \pm Z)$ and $(X,Y,Z)$ has uniform$[-1,1]$ marginals. 
First, by \eqref{circcopformula} and \eqref{psiantisymmetry}, the complementary 
two-dimensional marginal cdf is given on all of $C_2$ by
\begin{equation}\label{compcdfB2}
\bar F(x,y)=F(-x,-y)=\frac{1-x-y}{4}+\alpha(x,y),\qquad (x,y)\in C_2.
\end{equation}
\noindent{\it Case 1a:} $0\le x,y,z$, $x^2+y^2+z^2<1$. 
By inclusion-exclusion, \eqref{compcdfB2}, and \eqref{Fbarxyznice},
\begin{eqnarray*}
\lefteqn{F(x,y,z)}\\
 &=&1-P[X>x]-P[Y>y]-P[Z>z]+P[X>x,\,Y>y]+P[X>x,\,Z>z]\\
  &&+\ P[Y>y,\,Z>z]-P[X>x,\,Y>y,\,Z>z]\\
 &=&1-\frac{1-x}{2}-\frac{1-y}{2}-\frac{1-z}{2}+\frac{1-x-y}{4}+\alpha(x,y)+\frac{1-x-z}{4}+\alpha(x,z)\\
  &&+\ \frac{1-y-z}{4}+\alpha(y,z)-\frac{1- x-y-z}{8}  - \frac{\Delta(x,y,z)}{2}\\
   &=&\frac{1+x+y+z}{8}+\frac{\Delta(x,y,z)}{2}.
\end{eqnarray*}
\noindent{\it Case 1b:} $0\le x,y,z$, $x^2+y^2+z^2\ge1$. Here $P[X>x,\,Y>y,\,Z>z]=0$, so
\begin{eqnarray*}
\lefteqn{F(x,y,z)}\\
 &=&1-P[X>x]-P[Y>y]-P[Z>z]+P[X>x,\,Y>y]+P[X>x,\,Z>z]\\
  &&+\ P[Y>y,\,Z>z]\\
 &=&1-\frac{1-x}{2}-\frac{1-y}{2}-\frac{1-z}{2}+\frac{1-x-y}{4}+\alpha(x,y)+\frac{1-x-z}{4}+\alpha(x,z)\\
  &&+\ \frac{1-y-z}{4}+\alpha(y,z)\\
   &=&\frac{1+x+y+z}{8}+\frac{\Delta(x,y,z)}{2}+\ \sigma(xyz)\left[\frac{1-|x|-|y|-|z|}{8}+\frac{\Delta(|x|,|y|,|z|)}{2}\right].
\end{eqnarray*}
\noindent{\it Case 2a:} $x\le0\le y,z$, $x^2+y^2+z^2<1$. 
By inclusion-exclusion, \eqref{compcdfB2}, and \eqref{Fbarxyznice},
\begin{eqnarray*}
\lefteqn{F(x,y,z)}\\
&=&P[X\le x]-P[X\le x,\,Y>y]-P[X\le x,\,Z>z]+P[X\le x,\,Y>y,\,Z>z]\\
&=&\frac{x+1}{2}-P[-X\le x,\,Y>y]-P[-X\le x,\,Z>z]+P[-X\le x,\,Y>y,\,Z>z]\\
&=&\frac{x+1}{2}-P[X\ge-x,\,Y>y]-P[X\ge-x,\,Z>z]+P[X\ge-x,\,Y>y,\,Z>z]\\
  &=&\frac{x+1}{2} -\frac{1+x-y}{4}-\alpha(-x,y) -\frac{1+x-z}{4}-\alpha(-x,z)\\
   &&+\ \frac{1+x-y-z}{8}  + \frac{\Delta(-x,y,z)}{2}\\
   &=&\frac{1+x+y+z}{8}+\frac{\Delta(x,y,z)}{2}.
\end{eqnarray*}
\noindent{\it Case 2b:} $x\le0\le y,z$, $x^2+y^2+z^2\ge1$.  Here $P[X\ge-x,\,Y>y,\,Z>z]=0$, so
\begin{eqnarray*}
\lefteqn{F(x,y,z)}\\
 &=&P[X\le x]-P[X\le x,\,Y>y]-P[X\le x,\,Z>z]\\
&=&\frac{x+1}{2}-P[X\ge-x,\,Y>y]-P[X\ge-x,\,Z>z]\\
  &=&\frac{x+1}{2} -\frac{1+x-y}{4}-\alpha(-x,y) -\frac{1+x-z}{4}-\alpha(-x,z)\\
   &=&\frac{1+x+y+z}{8}+\frac{\Delta(x,y,z)}{2}+\ \sigma(xyz)\left[\frac{1-|x|-|y|-|z|}{8}+\frac{\Delta(|x|,|y|,|z|)}{2}\right].
\end{eqnarray*}
\noindent{\it Case 3a:} $y\le0\le x,z$, $x^2+y^2+z^2<1$. Similar to Case 2a.
\smallskip

\noindent{\it Case 3b:} $y\le0\le x,z$, $x^2+y^2+z^2\ge1$. Similar to Case 2b.
\smallskip

\noindent{\it Case 4a:} $z\le0\le x,y$, $x^2+y^2+z^2<1$. Similar to Case 2a.
\smallskip

\noindent{\it Case 4b:} $z\le0\le x,y$, $x^2+y^2+z^2\ge1$. Similar to Case 2b.
\medskip

\noindent{\it Case 5a:} $x,y\le0\le z$, $x^2+y^2+z^2<1$. 
By inclusion-exclusion, \eqref{compcdfB2}, and \eqref{Fbarxyznice},
\begin{eqnarray*}
F(x,y,z)
&=&P[X\le x,\,Y\le y]-P[X\le x,\,Y\le y,\,Z>z]\\
&=&P[X\ge -x,\,Y\ge -y]-P[X\ge -x,\,Y\ge -y,\,Z>z]\\
&=&\frac{1+x+y}{4}+\alpha(-x,-y)-\frac{1+x+y-z}{8} -\frac{\Delta(-x,-y,z)}{2}\\
   &=&\frac{1+x+y+z}{8}+\frac{\Delta(x,y,z)}{2}.
\end{eqnarray*}
\noindent{\it Case 5b:} $x,y\le0\le z$, $x^2+y^2+z^2\ge1$. Here $P[X\ge-x,\,Y\ge-y,\,Z>z]=0$, so
\begin{eqnarray*}
\lefteqn{F(x,y,z)}\\
&=&P[X\le x,\,Y\le y]\\
&=&P[X\ge -x,\,Y\ge -y]\\
&=&\frac{1+x+y}{4}+\alpha(-x,-y)\\
&=&\frac{1+x+y+z}{8}+\frac{\Delta(x,y,z)}{2}+\ \sigma(xyz)\left[\frac{1-|x|-|y|-|z|}{8}+\frac{\Delta(|x|,|y|,|z|)}{2}\right].
\end{eqnarray*}
\noindent{\it Case 6a:} $x,z\le0\le y$, $x^2+y^2+z^2<1$. Similar to Case 5a.
\smallskip

\noindent{\it Case 6b:} $x,z\le0\le y$, $x^2+y^2+z^2\ge1$. Similar to Case 5b.
\smallskip

\noindent{\it Case 7a:} $y,z\le0\le x$, $x^2+y^2+z^2<1$. Similar to Case 5a.
\smallskip

\noindent{\it Case 7b:} $y,z\le0\le x$, $x^2+y^2+z^2\ge1$. Similar to Case 5b.
\medskip

\noindent{\it Case 8a:} $x,y,z\le0\le z$, $x^2+y^2+z^2<1$.
\begin{eqnarray*}
F(x,y,z)
&=&P[X\le x,\,Y\le y,\,Z\le z]\\
&=&P[X\ge -x,\,Y\ge -y,\,Z\ge-z]\\
&=&\frac{1+x+y+z}{8} +\frac{\Delta(-x,-y,-z)}{2}\\
   &=&\frac{1+x+y+z}{8}+\frac{\Delta(x,y,z)}{2}.
\end{eqnarray*}
\noindent{\it Case 8b:} $x,y,z\le0$, $x^2+y^2+z^2\ge1$.
\begin{eqnarray*}
\lefteqn{F(x,y,z)=0}\\
&=&\frac{1+x+y+z}{8}+\frac{\Delta(x,y,z)}{2}+\ \sigma(xyz)\left[\frac{1-|x|-|y|-|z|}{8}+\frac{\Delta(|x|,|y|,|z|)}{2}\right].
\end{eqnarray*}
\end{proof}

\section{A One-parameter Family of Elliptical Copulas}
\label{sec:ellipticalcopula}

Let $(X,Y)\sim f(x,y)$ in \eqref{fcirc}, the unique circularly symmetric 
distribution on the unit disk $B_2$ with uniform$[-1,1]$ marginals. 
For any angle $\gamma\in(-\pi/2, \pi/2)$, consider the transformed variables
\begin{equation}\label{gammarotation}
U=X,\quad V_\gamma=X\sin\gamma + Y\cos\gamma.
\end{equation}
By the circular symmetry of $(X,Y)$, $V_\gamma\eqd Y\sim\mathrm{uniform}[-1,1]$, 
so the random vector $(U,V_\gamma)$ again generates a copula on the centered square $C_2$. 
Denote the pdf and cdf of $(U,V_\gamma)$ by $f_\gamma(u,v)$ and $F_\gamma(u,v)$ respectively. 
Then $\{F_\gamma\mid\gamma\in(-\pi/2, \pi/2)\}$ is a one-parameter family of 
\emph{elliptical copulas}, so-called because the support of $(U,V_\gamma)$ is the ellipse
\begin{equation}\label{Egamma}
E_\gamma:=\{(u,v)\mid u^2 + v^2 - 2uv \sin\gamma\le \cos^2 \gamma \}.
\end{equation}
(Note that $E_0=B_2$.) From \eqref{gammarotation}, the 
correlation coefficient of $U$ and $V_\gamma$ is given simply by
\begin{equation}\label{gammacorr}
\rho(U,V_\gamma)=\sin\gamma,
\end{equation}
so $\gamma$ indicates the degree of linear dependence between $U$ and $V_\gamma$.

\begin{prop}\label{fgammapdf} The pdf  of $(U,V_\gamma)$ is given by
\begin{equation}\label{fgam}
f_\gamma(u,v)= \frac{1}{2 \pi \sqrt{ \cos^2 \gamma - u^2 - v^2 + 2uv \sin \gamma } }1_{E_\gamma}(u,v).
\end{equation}
\end{prop} 
\begin{proof} The pdf can be obtained by a standard Jacobian computation. From \eqref{gammarotation},
\begin{equation}\label{xuvyuv}
x(u,v) = u,\quad y(u,v) = \frac{v - u\sin \gamma}{\cos(\gamma)}, 
\end{equation}
so 
\begin{eqnarray*}
\begin{array}{l l}
 \frac{\partial x(u,v)}{\partial u}  = 1, \qquad  & \frac{\partial x(u,v)}{\partial v} = 0, \\
 \frac{\partial y(u,v)}{\partial u} = -\tan \gamma,  & \frac{\partial y(u,v)}{\partial v} = \frac{1}{\cos \gamma} .
 \end{array}
\end{eqnarray*}
Thus the Jacobian of the transformation is $J=1/\cos \gamma $,  so from \eqref{fcirc} we obtain
\begin{eqnarray*}
f_\gamma (u,v) & = & f (x(u,v), y(u,v)) \cdot J \\
& = & \frac{1}{2 \pi \sqrt{1 - u^2 - (v-u\sin\gamma)^2/\cos^2 \gamma} } 1_{B_2} (x(u,v), y(u,v)) \cdot \frac{1}{\cos \gamma}\\
& = & \frac{\cos \gamma}{2 \pi \sqrt{(1-u^2)\cos^2 \gamma  - (v - u \sin \gamma)^2}} \frac{1}{\cos \gamma} \cdot
          1_{B_2} (x(u,v), y(u,v)) \\
& = & \frac{1}{2 \pi \sqrt{ \cos^2 \gamma - u^2 - v^2 +2uv \sin \gamma } }
          1_{E_\gamma}(u,v).
\end{eqnarray*}
\end{proof}

Figure~\ref{fig5.1} 
shows the density $f_{\gamma}(u,v)$ with $\gamma = \pi/4$.

\begin{figure}[h]
\centering
\includegraphics[scale=1]{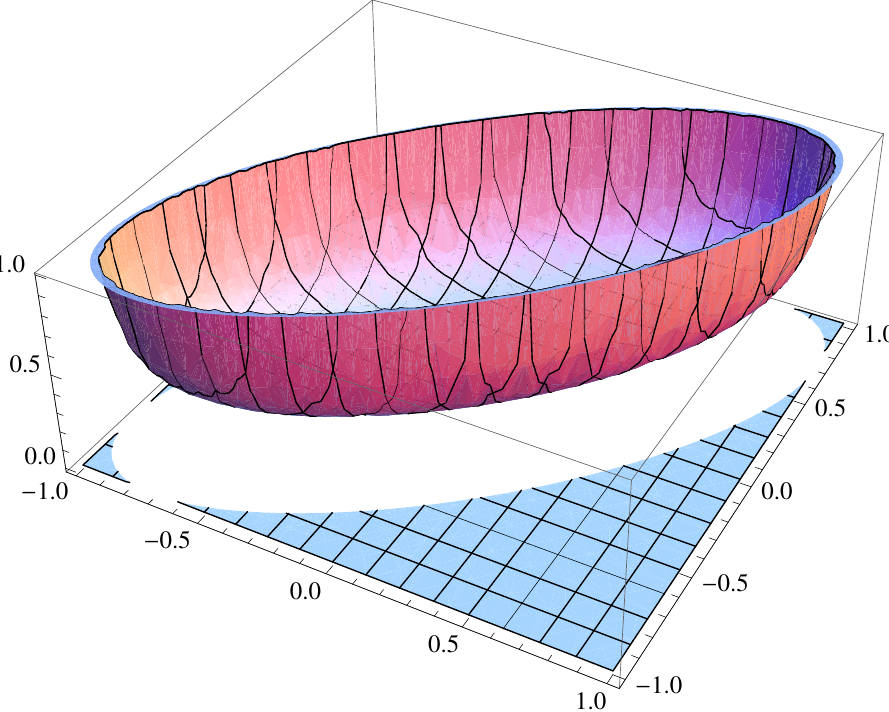}
\caption{The density $f_{\gamma}(u,v)$ (with $\gamma = \pi/4$)} 
\label{fig5.1}
\end{figure}

To describe the family of elliptical copulas $F_\gamma$, 
we extend the definitions \eqref{psi} and \eqref{psiextension} as follows. 
First, for $(u,v)\in E_\gamma$ define
\begin{eqnarray}
\alpha_\gamma(u,v)&=&
\frac{1}{2\pi}\left[u\arcsin \left ( \frac{v-u\sin\gamma}{\cos\gamma\sqrt{1-u^2}} \right ) 
         +v\arcsin\left ( \frac{u-v\sin\gamma}{\cos\gamma\sqrt{1-v^2}} \right )\right.\nonumber\\
&& \left.\qquad-\arcsin\left ( \frac{uv-\sin\gamma}{ \sqrt{(1-u^2)(1-v^2)}} \right )\right].\label{psigamma}
\end{eqnarray}
Note that $\alpha_\gamma$ reduces to $\alpha$ in \eqref{psi} when $\gamma=0$, i.e., when $V_\gamma = Y$. From \eqref{areacapintersect},
\begin{equation}\label{Apsigamma}
A(\arccos u,\,\arccos v;\,\frac{\pi}{2}-\gamma)=(1-u-v)\pi+4\pi\alpha_\gamma(u,v).
\end{equation}

Next, extend the definition of $\alpha_\gamma(u,v)$ to $C_2\setminus E_\gamma$ as follows (see Figure~\ref{fig5.2}):
\begin{equation}\label{psigammaextension}
\alpha_\gamma(u,v)=
\begin{cases}
\frac{u+v-1}{4}& \mathrm{if}\ (u,v)\in R_5(\gamma):=(C_2\setminus E_\gamma) \cap\{(u,v)\mid u+v>1+\sin\gamma\},\\
\frac{u-v+1}{4}& \mathrm{if}\ (u,v)\in R_6(\gamma):=(C_2\setminus E_\gamma) \cap\{(u,v)\mid v-u>1-\sin\gamma\},\\
\frac{-u+v+1}{4}& \mathrm{if}\ (u,v)\in R_7(\gamma):=(C_2\setminus E_\gamma) \cap\{(u,v)\mid v-u<\sin\gamma-1\},\\
\frac{-u-v-1}{4}& \mathrm{if}\ (u,v)\in R_8(\gamma):=(C_2\setminus E_\gamma) \cap\{(u,v)\mid u+v<-\sin\gamma-1\}.
\end{cases}
\end{equation}

Note that \eqref{psigamma} and \eqref{psigammaextension} agree on 
$\partial E_\gamma$, i.e., when $u^2 + v^2 - 2uv \sin\gamma=\cos^2 \gamma$. 
Also note that \eqref{psigammaextension} reduces to $\alpha$ in \eqref{psiextension} 
when $\gamma=0$. The following lemma will be useful for the proof of Theorem \ref{Fgamcdf}.

\begin{lem}\label{Fgammasymm}
Let $(U,V)$ be a bivariate random vector in $C_2$ with uniform$[-1,1]$ 
marginals that satisfies $(U,V)\eqd(-U,-V)$. Then the cdf $F(u,v)$ satisfies
\begin{equation}\label{FuvFminusuv}
F(u,v)=\frac{u+v}{2}+F(-u,-v),\quad (u,v)\in C_2.
\end{equation}
\end{lem}
\begin{proof} By the symmetry condition,
\begin{eqnarray*}
F(u,v)
&=&P[-U\le u,\,-V\le v]
        =P[U\ge-u,\,V\ge-v]\\
 &=&1-P[U<-u]-P[V<-v]+P[U<-u,\,V<-v]\\
  &=&1-\left(\frac{-u+1}{2}\right)+\left(\frac{-v+1}{2}\right)+P[U\le-u,\,V\le-v]\\
   &=&\frac{u+v}{2}+F(-u,\,-v).
\end{eqnarray*}
\end{proof}

\begin{figure}[h]
\centering
\includegraphics[scale=1]{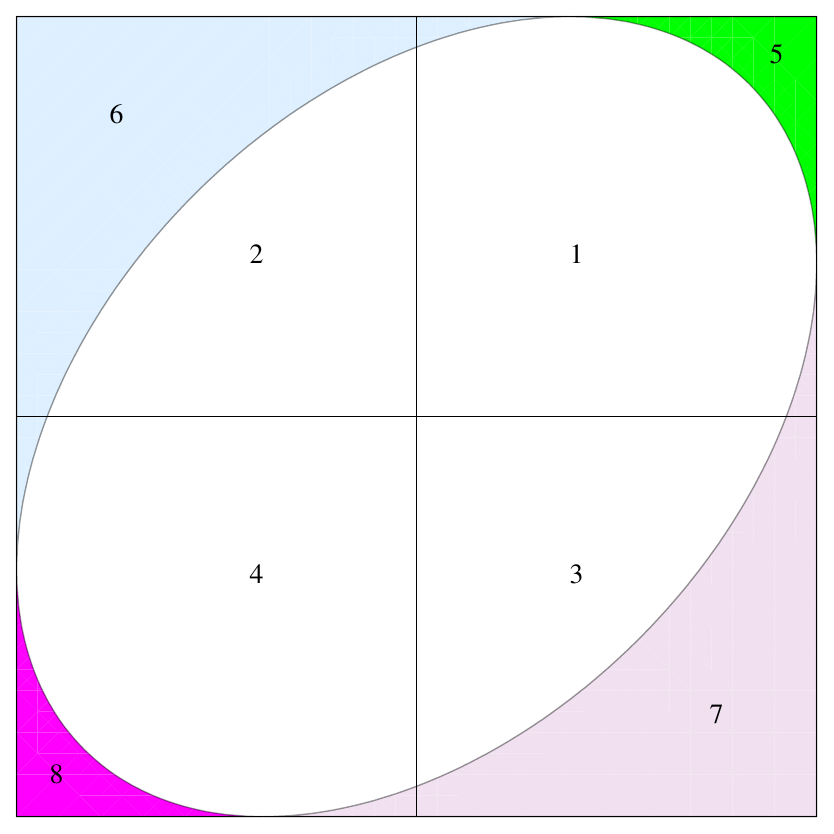}
\caption{Eight regions $R_1 (\gamma) - R_8 (\gamma)$ for an elliptical copula (with $\gamma = \pi/8$)} 
\label{fig5.2}
\end{figure}

\begin{thm}\label{Fgamcdf} The  cdf $\equiv$ copula of $(U,V_\gamma)$ is given by
\begin{equation}\label{Fgam}
F_\gamma(u,v)=\frac{u+v+1}{4}+\alpha_\gamma(u,v),\qquad\qquad(u,v)\in C_2.
\end{equation}
\end{thm} 

\begin{proof}  
 To find $F_\gamma(u,v)$ we again use the formula \eqref{areacapintersect} for the 
 area of the intersection of two spherical caps on $\partial B_3$. 
 Here, unlike \eqref{AA}, the axes of the two caps are not necessarily perpendicular. 
 The single formula \eqref{Fgam} is obtained by considering the partition 
 $C_2=\cup_{i=1}^8R_i(\gamma)$, where $R_5(\gamma)-R_8(\gamma)$ are defined in 
 \eqref{psigammaextension} and (see Figure~\ref{fig5.2})
 \begin{eqnarray*}
 R_1(\gamma)&=&E_\gamma\cap\{(u,v)\mid 0\le u,v\le1\},\\
  R_2(\gamma)&=&E_\gamma\cap\{(u,v)\mid-1\le u\le0\le v\le1\},\\
   R_3(\gamma)&=&E_\gamma\cap\{(u,v)\mid-1\le v\le0\le u\le1\},\\
    R_4(\gamma)&=&E_\gamma\cap\{(u,v)\mid-1\le u,v\le0\}.
 \end{eqnarray*}

\noindent {\it Case 1:} $(u,v)\in R_1(\gamma)$.   By using Figure~\ref{fig5.3},

\begin{figure}[h]
\centering
\includegraphics[scale=1]{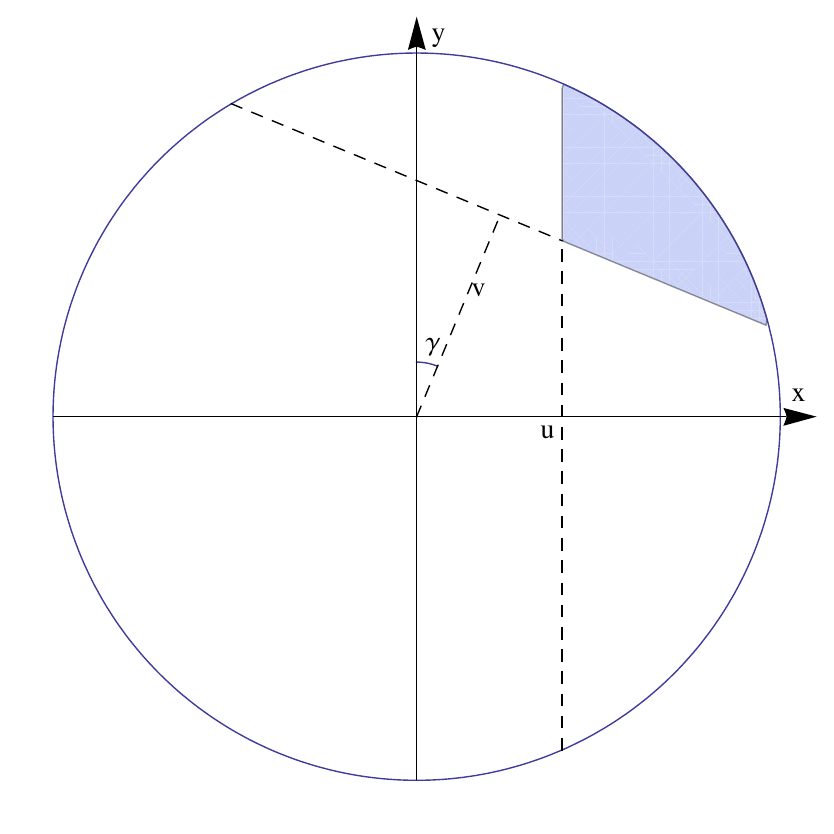}
\caption{The region $[X>u, X\sin(\gamma) + Y\cos(\gamma) >v]$ for Case 1 (with $\gamma = \pi/8$)} 
\label{fig5.3}
\end{figure}

\begin{eqnarray*}
F_\gamma(u,v)&=&P[U\le u,\,V_\gamma\le v]\\
&=&1-P[U>u]-P[V_\gamma>v]+P[U>u,\,V_\gamma>v]\\
&=&1-\left(\frac{1-u}{2}\right)-\left(\frac{1-v}{2}\right)+P[X>u,\,X\sin\gamma+Y\cos\gamma>v]\\
 &=&\frac{u+v}{2}+\frac{1}{4\pi}A(\arccos u,\arccos v;\,\pi/2-\gamma)\\
 &=& \frac{u+v+1}{4}+\alpha_\gamma(u,v)\qquad\qquad\qquad\qquad\qquad\mathrm{[by\ \eqref{Apsigamma}]}.
\end{eqnarray*}

\noindent {\it Case 2:}   $(u,v))\in R_2(\gamma)$.  Because $(X,Y)\eqd(-X,Y)$ and using Figure~\ref{fig5.3-2}
\begin{eqnarray*}
F_\gamma(u,v)&=&P[U\le u]-P[U\le u,\,V_\gamma>v]\\
&=&\frac{u+1}{2}-P[X\le u,\,X\sin\gamma+Y\cos\gamma>v]\\
&=&\frac{u+1}{2}-P[-X\le u,\,-X\sin\gamma+Y\cos\gamma>v]\\
&=&\frac{u+1}{2}-P[X\ge -u,\,X\sin(-\gamma)+Y\cos(-\gamma)>v]\\
&=&\frac{u+1}{2}-\frac{1}{4\pi}A(\arccos(-u),\arccos v;\,\pi/2+\gamma)\\
&=& \frac{u+v+1}{4}-\alpha_{-\gamma}(-u,v)\ \ \qquad\qquad\qquad\qquad\mathrm{[by\ \eqref{Apsigamma}]}\\
&=& \frac{u+v+1}{4}+\alpha_\gamma(u,v)\qquad\qquad\qquad\qquad\qquad\mathrm{[by\ \eqref{psigamma}]}.
\end{eqnarray*}

\begin{figure}[h]
\centering
\includegraphics[scale=1]{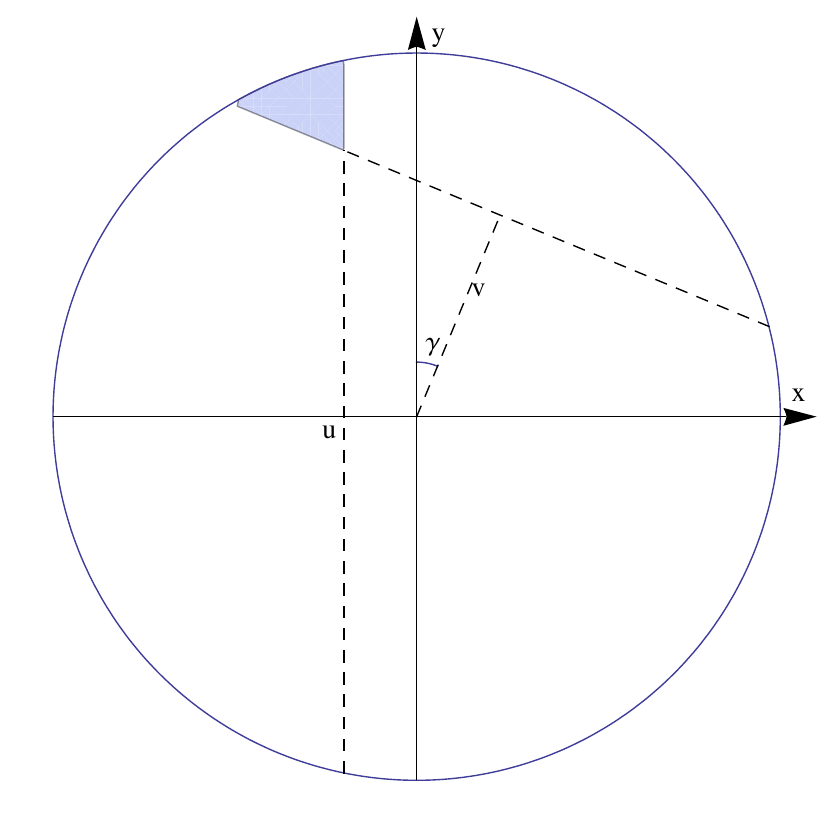}
\caption{The region $[X<u, X\sin(\gamma) + Y\cos(\gamma) >v]$ for Case 2 (with $\gamma = \pi/8$)} 
\label{fig5.3-2}
\end{figure}

\noindent {\it Case 3:}   $(u,v)\in R_3(\gamma)$. Then $(-u,-v)\in R_2(\gamma)$, so by Lemma \ref{Fgammasymm} and Case 2,
\begin{eqnarray*}
F_\gamma(u,v)&=&\frac{u+v}{2}+F_\gamma(-u,-v)\\
 &=&\frac{u+v}{2}+\frac{-u-v+1}{4}+\alpha_\gamma(-u,-v)\\
 &=& \frac{u+v+1}{4}+\alpha_\gamma(u,v)\qquad\qquad\qquad\qquad\mathrm{[by\ \eqref{psigamma}]}.
\end{eqnarray*}

\noindent {\it Case 4:}  $(u,v)\in R_4(\gamma)$. Then $(-u,-v)\in R_1(\gamma)$, so by Lemma \ref{Fgammasymm} and Case 1,
the argument for Case 3 applies verbatim.
\bigskip

\noindent {\it Case 5:}  $(u,v)\in R_5(\gamma)$.
\begin{eqnarray*}
F_\gamma(u,v)&=&1-P[U>u]-P[V_\gamma>v]\\
&=&1-\left(\frac{1-u}{2}\right)-\left(\frac{1-v}{2}\right)\\
&=&\frac{u+v+1}{4}+\alpha_\gamma(u,v)\qquad\qquad \qquad[\mathrm{by\ \eqref{psigammaextension}}].
\end{eqnarray*}

\noindent {\it Case 6:}  $(u,v)\in R_6(\gamma)$.
\begin{eqnarray*}
F_\gamma(u,v)&=&P[U\le u]\\
 &=&\frac{u+1}{2}\\
&=&\frac{u+v+1}{4}+\alpha_\gamma(u,v)\qquad\qquad \qquad[\mathrm{by\ \eqref{psigammaextension}}].
\end{eqnarray*}

\noindent {\it Case 7:}  $(u,v)\in R_7(\gamma)$. Then $(-u,-v)\in R_6(\gamma)$, 
so by Lemma \ref{Fgammasymm} and Case 6, the argument for Case 3 applies verbatim.
\bigskip

\noindent {\it Case 8:} $(u,v)\in R_8(\gamma)$. Then $(-u,-v)\in R_5(\gamma)$, 
so by Lemma \ref{Fgammasymm} and Case 5, the argument for Case 3 applies verbatim.
\end{proof}

\begin{figure}[h]
\centering
\includegraphics[scale=1]{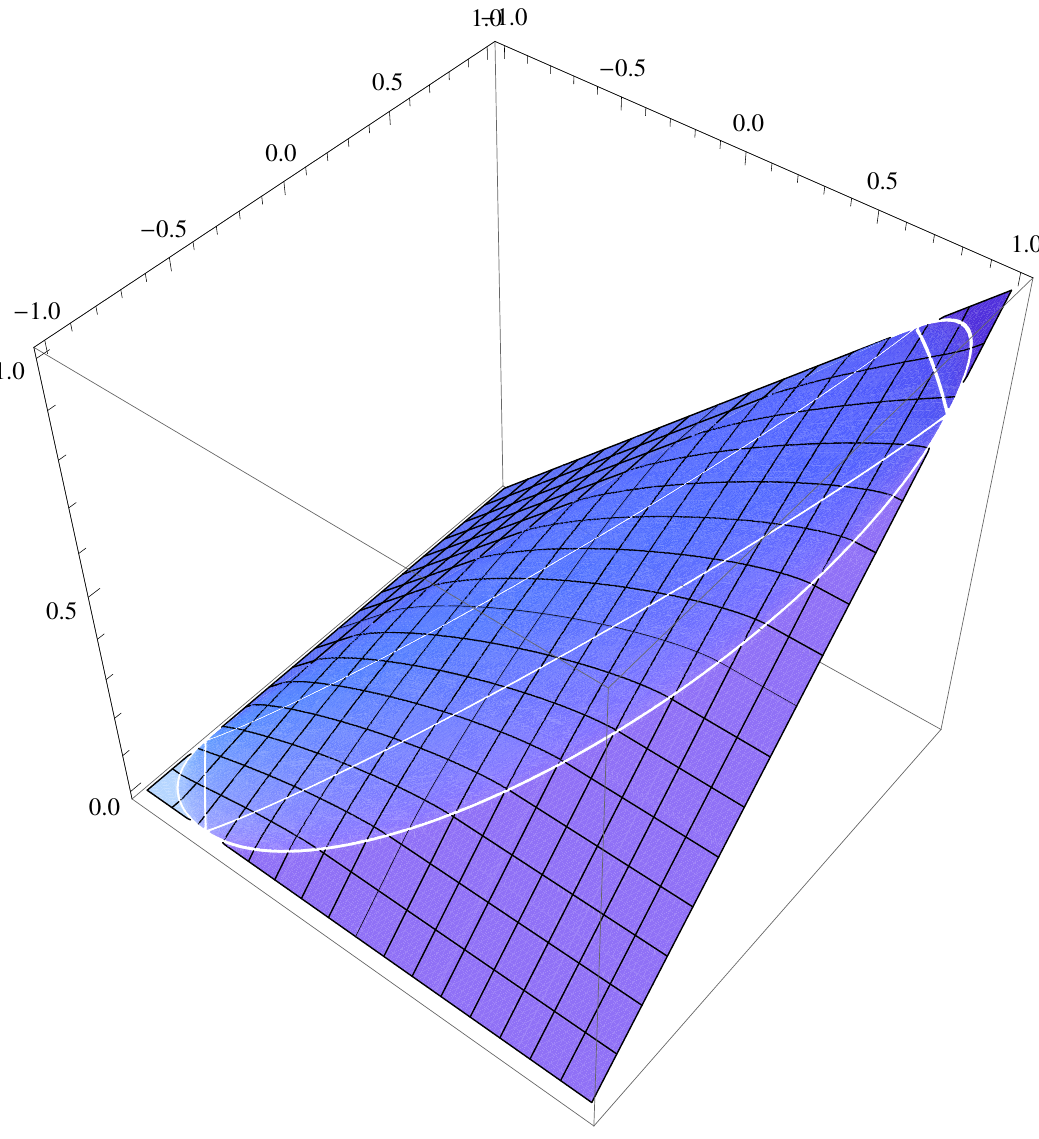}
\caption{The copula $F_{\gamma}(u,v)$ (with $\gamma = \pi/4$)} 
\label{fig5.4}
\end{figure}

\bigskip

\section{Copulas Derived from the Uniform Distribution on the Unit Ball}
\label{sec:nonlinearcopula}

Up to now we have addressed the question of whether copulas can be 
generated by means of linear functions of  a circularly symmetric or 
spherically symmetric random vector. Now we ask whether non-linear 
functions of such random vectors can generate copulas. 
We shall restrict attention to random vectors uniformly distributed 
over the unit ball $B_d$ and produce relatively simple non-linear functions that generate copulas on $C_d$.

We begin with the bivariate case. Suppose that $(X,Y) $ is distributed 
uniformly on the unit disk $B_2 = \{ (x,y) \in \RR^2\mid \ x^2 + y^2 \le 1 \}$.   Because 
\begin{eqnarray*}
X\mid Y&\sim&\mathrm{uniform}[-\sqrt{1-Y^2},\,\sqrt{1-Y^2}]\\
\mathrm{and}\qquad Y\mid X&\sim&\mathrm{uniform}[-\sqrt{1-X^2},\,\sqrt{1-X^2}],
\end{eqnarray*}
it follows that the random variables
\begin{equation}\label{UVdef}
U := \frac{X}{\sqrt{1- Y^2} },  \qquad V := \frac{Y}{\sqrt{1-X^2}}
\end{equation}
satisfy 
\begin{eqnarray*}
U\mid Y&\sim&\mathrm{uniform}[-1,1]\\
 V\mid X&\sim&\mathrm{uniform}[-1,1].
\end{eqnarray*}
Thus, $U$ and $Y$ are independent, $V$ and $X$ are independent, and unconditionally,
\begin{eqnarray*}
U&\sim&\mathrm{uniform}[-1,1]\\
\qquad V&\sim&\mathrm{uniform}[-1,1],
\end{eqnarray*}
so the joint distribution of $(U,V)$ generates a copula $F_(u,v)$ 
on the centered cube $C_2\equiv[-1,1]^2$. 
Note that $U$ and $V$ are not linear functions of $(X,Y)$.
\medskip

\noindent{\bf Question 6:} 
Are $U$ and $V$ independent, and if not, what is the nature of their dependence?  
\smallskip

\noindent {\bf Answer 6.} Clearly $U$ and $V$ are uncorrelated, since $E(U)=E(V)=0$ and
\begin{equation}
E(UV)=E\left(\frac{XY}{\sqrt{(1-X^2)(1-Y^2)}}\right)=0,
\nonumber
\end{equation}
all by the circular symmetry of $(X,Y)$. However, the joint pdf and cdf of $(U,V)$ 
derived below show that they are not independent.
\medskip

\begin{prop}\label{Prop1}
 The joint density of $(U,V)$  is given by 
\begin{eqnarray}
f(u,v) = \frac{1}{\pi} \frac{\sqrt{(1-u^2)(1-v^2)}}{(1- u^2 v^2)^2} 1_{C_2} (u,v) .
\label{CopulaDensityOnMinusOneToOneSquare}
\end{eqnarray}
\end{prop}

\begin{figure}[h]
\centering
\includegraphics[scale=1]{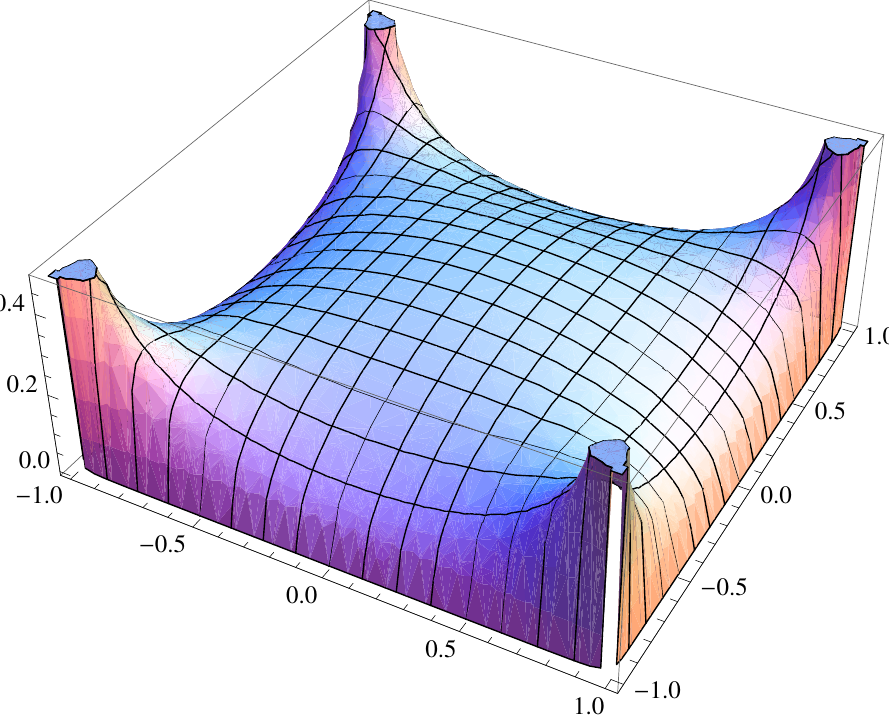}
\caption{Joint density $f(u,v)$ of $(U, V)$ in Proposition 6.1} 
\label{fig6.1}
\end{figure}

\begin{proof} This pdf is again obtained via the Jacobian method. It follows from \eqref{UVdef} that 
\begin{eqnarray*}
&& u^2 (1-y^2) = x^2, \qquad \mbox{and} \qquad  v^2 (1-x^2) = y^2 .
\end{eqnarray*}
Substitution of the second expression for $y^2$ into the left side of the first relation and vice-versa
yields
\begin{eqnarray*}
x^2 = \frac{u^2 (1-v^2)}{1 - u^2 v^2}, \qquad y^2 = \frac{v^2(1-u^2)}{1 - u^2 v^2} ,
\end{eqnarray*}
so, since $x$ and $u$ ($y$ and $v$) have the same signs by \eqref{UVdef}, we obtain
\begin{equation}\label{xyuv}
x\equiv x(u,v)= \frac{u \sqrt{1-v^2}}{\sqrt{1 - u^2 v^2}}, \qquad  y\equiv y(u,v) = \frac{v \sqrt{1-u^2}}{\sqrt{1 - u^2 v^2}} .
\end{equation}
Thus
\begin{eqnarray*}
\frac{\partial x}{\partial u} 
& = & \sqrt{1-v^2}\left[\frac{1}{\sqrt{1-u^2v^2}} + u  (1- u^2 v^2)^{-3/2} (u v^2)\right] \\
& = & \frac{\sqrt{1-v^2}}{\sqrt{1-u^2 v^2}} \left \{ 1 + \frac{u^2 v^2}{1 - u^2 v^2} \right \} \\
& = & \frac{\sqrt{1-v^2}}{(1-u^2 v^2)^{3/2}},
\end{eqnarray*} 
\begin{eqnarray*}
\frac{\partial x}{\partial v} 
& = & u\left[\frac{  (1-v^2)^{-1/2} (-v)}{\sqrt{1-u^2v^2}} +  \sqrt{1-v^2}(1- u^2 v^2)^{-3/2} (u^2 v) \right]\\
& = & \frac{uv}{\sqrt{1-v^2} \sqrt{1-u^2 v^2}} \left \{ -1 + \frac{u^2 (1-v^2)}{1 - u^2 v^2} \right \}  \\
& = & - \frac{uv (1-u^2)}{\sqrt{1-v^2} (1 - u^2 v^2)^{3/2} } .
\end{eqnarray*} 
By symmetry it follows that the Jacobian is given by 
\begin{eqnarray*}
J = \left ( \begin{array}{ c c}   \frac{\sqrt{1-v^2}}{(1-u^2 v^2)^{3/2}} & - \frac{uv (1-u^2)}{\sqrt{1-v^2} (1 - u^2 v^2)^{3/2} } \\
- \frac{uv (1-v^2)}{\sqrt{1-u^2} (1 - u^2 v^2)^{3/2} } & \frac{\sqrt{1-u^2}}{(1-u^2 v^2)^{3/2}}  
\end{array} \right ) ,
\end{eqnarray*}
and hence the determinant of $J$ is given by 
\begin{eqnarray*}
| J | = \frac{\sqrt{(1-u^2)(1-v^2)}}{(1 - u^2 v^2)^2} .
\end{eqnarray*}
Because the pdf of $(X,Y)$ is $f(x,y)=\frac{1}{\pi}1_{B_2}(x,y)$, the result  \eqref{CopulaDensityOnMinusOneToOneSquare} follows.
\end{proof}

For $0\le u,v\le1$, $(u,v)\ne(1,1)$, let $E_1(u)$ and $E_2(v)$ be the ellipses
\begin{eqnarray}\label{twoellipses}
E_1(u)&=&\left\{(x,y)\Bigm| \frac{x^2}{u^2}+y^2\le1 \right\},\\
E_2(v)&=&\left\{(x,y)\Bigm| x^2+\frac{y^2}{v^2}\le1 \right\}.
\end{eqnarray}
The next  lemma leads to the cdf $F(u,v)$ corresponding to the pdf \eqref{CopulaDensityOnMinusOneToOneSquare}.
\begin{lem}\label{AreaE1E2}
\begin{equation*}
 \mathrm{Area}(E_1(u)\cap E_2(v))
 =2u\arcsin\left( \frac{v \sqrt{1-u^2}}{\sqrt{1 - u^2 v^2}}\right)+2v\arcsin\left(\frac{u \sqrt{1-v^2}}{\sqrt{1 - u^2 v^2}}\right).
\end{equation*}
\end{lem}
\begin{proof}Define the points $o,a,b,c,d,d,f,g$ as follows:  
see Figure~\ref{fig6.2}  

\begin{figure}[h]
\centering
\includegraphics[scale=1]{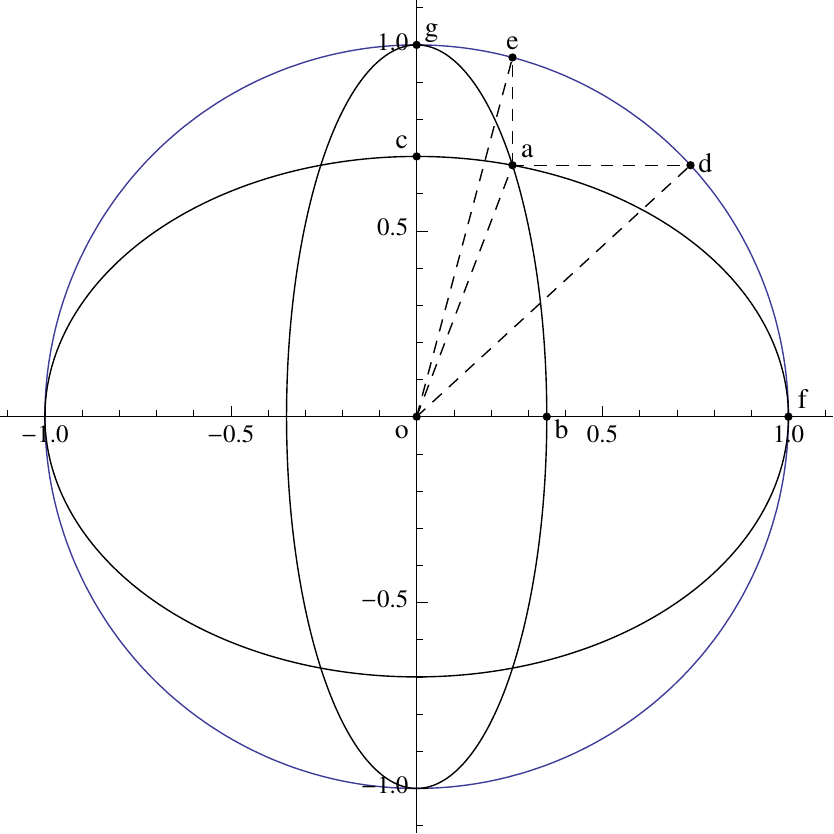}
\caption{Integration regions for Lemma 6.2} 
\label{fig6.2}
\end{figure}

\begin{eqnarray*}
o&=&(0,0),\\
a&=&(x(u,v),\,y(u,v))=\left(\frac{u \sqrt{1-v^2}}{\sqrt{1 - u^2 v^2}},\,\frac{v \sqrt{1-u^2}}{\sqrt{1 - u^2 v^2}}\right),\\
b&=&(u,0),\\
c&=&(0,v),\\
d&=&(\sqrt{1-y^2(u,v)},\,y(u,v)),\\
e&=&(x(u,v),\,\sqrt{1-x^2(u,v)}),\\
f&=&(1,0),\\
g&=&(0,1).
\end{eqnarray*}
Then
\begin{eqnarray*}
\frac{1}{4}\mathrm{Area}(E_1(u)\cap E_2(v))&=&\mathrm{Area}(oab)+\mathrm{Area}(oac)\\
 &=&u\,\mathrm{Area}(odf)+v\,\mathrm{Area}(oeg)\\
 &=&\frac{u}{2}\arcsin(y(u,v))+ \frac{v}{2}\arcsin(x(u,v)),
\end{eqnarray*}
from which the result follows.
\end{proof}

\begin{thm}\label{secondcopula}
The copula (= cdf)  corresponding to the pdf \eqref{CopulaDensityOnMinusOneToOneSquare} is given by
\begin{equation*}
F(u,v)=\frac{u+v+1}{4}+\frac{u}{2\pi}
    \arcsin\left( \frac{v \sqrt{1-u^2}}{\sqrt{1 - u^2 v^2}}\right)
    +\frac{v}{2\pi}\arcsin\left(\frac{u \sqrt{1-v^2}}{\sqrt{1 - u^2 v^2}}\right),\ \ (u,v)\in C_2.
\label{NonLinearCopula1}
\end{equation*}
\end{thm}

\begin{figure}[h]
\centering
\includegraphics[scale=1]{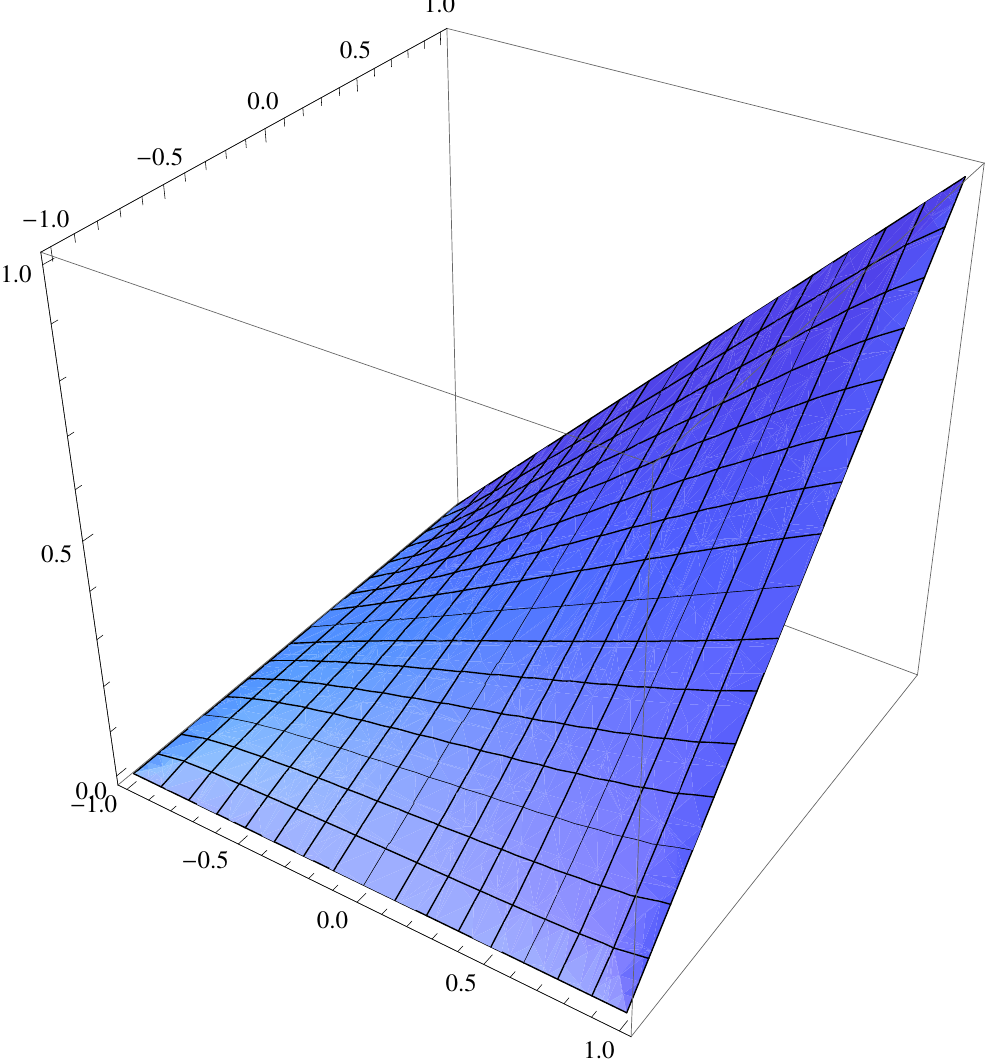}
\caption{Nonlinear transformation copula $F(u,v)$ in Theorem 6.3} 
\label{fig6.3}
\end{figure}

\begin{proof} 
Because $(U,V)$ is sign-change invariant and has uniform$[-1,1]$ marginals, 
it follows from \eqref{F0} and \eqref{XYabsXY} in Lemma \ref{absvalues}  
and from \eqref{UVdef} that for $(u,v)\in C_2$,
\begin{eqnarray*}
F(u,v)&=&\frac{u+v+1}{4}+\sigma(uv)\,P[0\le U\le |u|,\, 0\le V\le |v|]\label{UVabsUV}\\
&=&\frac{u+v+1}{4}+\frac{\sigma(uv)}{4}\,P[U^2\le u^2,\, V^2\le v^2]\label{UVabsUVsquared}\\
&=&\frac{u+v+1}{4}+\frac{\sigma(uv)}{4}\,P[(X,Y)\in E_1(u)\cap E_2(v)]\\
&=&\frac{u+v+1}{4}+\frac{\sigma(uv)}{4\pi}\mathrm{Area}(E_1(|u|)\cap E_2(|v|))\label{Buv}.
\end{eqnarray*}
The result now follows from Lemma \ref{AreaE1E2}.
\end{proof}
\medskip

\par\noindent
{\bf Problem:} 
The construction \eqref{UVdef} extends readily to generate a copula 
on $C_d$. For $d=3$, for example, let $(X,Y,Z)$ be uniformly distributed on the unit ball $B_3$ and define
\begin{equation}\label{UVW}
U:=\frac{X}{\sqrt{1-Y^2-Z^2}},\quad V:=\frac{Y}{\sqrt{1-X^2-Z^2}},\quad W:=\frac{Z}{\sqrt{1-X^2-Y^2}}.
\nonumber
\end{equation}
Then the marginal distributions of $U$, $V$, 
and $W$ are each uniform$[-1,1]$ so the cdf $G(u,v,w)$ is a copula on $C_3$.  
To find this copula one would need to determine 
\begin{equation}\label{area3}
\mathrm{Volume}(E_1(u)\cap E_2(v)\cap E_3(w)),
\nonumber
\end{equation}
where now, for $0\le u,v,w\le 1$,
$E_1(u)$, $E_2(v)$, and $E_3(w)$ are the ellipsoids
\begin{eqnarray*}
E_1(u)&=&\left\{(x,y,z)\Bigm| \frac{x^2}{u^2}+y^2+z^2\le1 \right\},\\
E_2(v)&=&\left\{(x,y,z)\Bigm| x^2+\frac{y^2}{v^2}+z^2\le1 \right\},\\
E_3(w)&=&\left\{(x,y,z)\Bigm| x^2+y^2+\frac{y^2}{w^2}\le1 \right\}.
\end{eqnarray*}

\noindent{\bf Acknowledgement:} 
We gratefully acknowledge several helpful suggestions by Ilya Vakser and Andrey Tovchigrechko.

\bigskip

\bibliographystyle{ims}
\bibliography{SphericalCopulas}

\end{document}